\documentclass[times,sort&compress,3p]{elsarticle}
\journal{Journal of Multivariate Analysis}
\usepackage[labelfont=bf]{caption}

\usepackage[utf8]{inputenc}

\usepackage{amsmath,amsfonts,amssymb,amsthm,booktabs,color,epsfig,graphicx,url}

\usepackage[colorlinks,citecolor=blue,urlcolor=blue,pdfauthor=author]{hyperref}
\usepackage{graphicx}
\usepackage{subcaption}

\usepackage[linesnumbered,ruled,vlined]{algorithm2e}

\PassOptionsToPackage{unicode}{hyperref}
\PassOptionsToPackage{naturalnames}{hyperref}

\usepackage{tikz}
\usepackage{pgfplots}
\usepgfplotslibrary{fillbetween}
\tikzset{>=latex}
\pgfplotsset{compat=1.16}
\usepackage{tkz-euclide}
\usepackage{xcolor}
\usetikzlibrary{calligraphy, calc, fit}
\usepackage{ifthen}
\usepackage{tabularx}
\usetikzlibrary{calc,positioning,matrix, decorations.pathreplacing}

\pgfdeclarelayer{bg}
\pgfsetlayers{bg,main}

\theoremstyle{plain}
\newtheorem{theorem}{Theorem}

\newtheorem{proposition}{Proposition}
\newtheorem{lemma}{Lemma}
\newtheorem{corollary}{Corollary}

\theoremstyle{definition}
\newtheorem{definition}{Definition}

\newtheorem{example}{Example}

\newcommand{\R}{{\mathbb{R}}}
\newcommand{\N}{{\mathbb{N}}}

\newcommand{\esp}{{\mathbb{E}}}
\newcommand{\risk}{{\mathcal{R}}}
\newcommand{\prob}{{\mathbb{P}}}
\newcommand{\pkg}[1]{\texttt{#1}}

\begin{document}

\begin{frontmatter}

\title{Functional linear regression with truncated signatures}

\author[1]{Adeline Fermanian \corref{mycorrespondingauthor}}

\address[1]{Sorbonne Université, CNRS, Laboratoire de Probabilités, Statistique et Modélisation, LPSM, F-75005 Paris, France \\
MINES ParisTech, PSL Research University, CBIO-Centre for Computational Biology, F-75006 Paris, France \\
Institut Curie, PSL Research University, F-75005 Paris, France \\
INSERM, U900, F-75005 Paris, France}

\cortext[mycorrespondingauthor]{Corresponding author. Email address: \url{adeline.fermanian@sorbonne-universite.fr}}

\begin{abstract}
We place ourselves in a functional regression setting and propose a novel methodology for regressing a real output on vector-valued functional covariates. This methodology is based on the notion of signature, which is a representation of a function as an infinite series of its iterated integrals. The signature depends crucially on a truncation parameter for which an estimator is provided, together with theoretical guarantees. An empirical study on both simulated and real-world datasets shows that the resulting methodology is competitive with traditional functional linear models, in particular when the functional covariates take their values in a high dimensional space.
\end{abstract}


\end{frontmatter}

\section{Introduction}
In a classical regression setting, a real output $Y$ is described by a finite number of predictors. A typical example would be to model the price of a house as a linear function of several characteristics such as surface area, number of rooms, location, and so on. These predictors are typically encoded as a vector in $\R^p$, $p \in \N^\ast$. However, some applications do not fall within this setting. For example, in medicine, a classical task consists of predicting the state of a patient (for example, ill or not) from the recording of several physiological variables over some time. The input data is then a function of time and not a vector. Similarly, sound recognition or stock market prediction tasks both consist of learning from time series, possibly multidimensional. Then, a natural idea is to extend the linear model to this more general setting, where one wants to predict from a functional input, of the form $X:[0,1] \to \R^d$, $d\geq 1$.

This casts our problem into the field of functional data analysis and more specifically within the framework of functional linear regression \citep{ramsay1991some,marx1999generalized}. This rich domain has undergone considerable developments in recent decades, as illustrated by the monographs of \citet{ramsay2005functional} and \citet{ferraty2006nonparametric}, and the review by \citet{morris2015functional}. One of the core principles of functional data analysis is to represent input functions on a set of basis functions, for example, splines, wavelets, or the Fourier basis. Another approach also consists in extracting relevant handcrafted features, depending on the field of application. For example, \cite{benzeghiba2007automatic} and \cite{turaga2008machine} provide overviews of learning methods specific to speech and human action recognition, respectively.

In this article, we build on the work of \citet{levin2013learning} and explore a novel approach to linear functional regression, called the signature linear model. Its main strength is that it is naturally adapted to vector-valued functions, which is not the case with most of the methods previously mentioned. Its principle is to represent a function by its signature, defined as an infinite series of its iterated integrals. Signatures date back from the 60s when \citet{chen1958integration} showed that a smooth path can be faithfully represented by its iterated integrals and it has been at the center of rough path theory in the 90s \citep{lyons2007differential,friz2010multidimensional}. Rough path theory has seen extraordinary developments in recent times, and, in particular, has gained attention from the machine learning community. Indeed, signatures combined with (deep) learning algorithms have been successfully applied in various fields, such as characters recognition \citep{yang2015chinese,yang2016deepwriterid,lai2017online,liu2017ps}, human action recognition \citep{li2017lpsnet,yang2017leveraging}, speech emotion recognition \citep{wang2019path}, medicine \citep{arribas2018signature,moore2019using,morrill2019sepsis,morrill2020utilization}, or finance \citep{arribas2020sig}. We refer the reader to \citet{chevyrev2016primer} for an introduction to signatures in machine learning, and to \citet{fermanian2021embedding} for a more recent overview.

We stress again that the main advantage of the signature approach is that it can handle multidimensional input functions, that is, functions $X:[0,1] \to \R^d$ where $d\geq 2$, whereas traditional methods were designed for real-valued functions. Many modern datasets come in this form with a large dimension $d$. Moreover, the signature method requires little assumptions on the regularity of $X$ and encodes nonlinear geometric information, that is, gives rise to interpretable regression coefficients. Finally, it is theoretically grounded by good approximation properties: any continuous function can be approximated arbitrarily well by a linear function of the truncated signature \citep{kiraly2016kernels}.

Since any continuous function of $X$ can be approximated by a linear function on its truncated signature, the estimation of a regression function boils down to the estimation of the coefficients in this scalar product. The truncation order of the signature is therefore a crucial parameter as it controls the complexity of the model. Thus, in our quest for a linear model on the signature, one of the main purposes of our article will be to estimate this parameter. With an estimator of the truncation order at hand, the methodology is complete and the signature linear model can be applied to both simulated and real-world data, demonstrating its good performance for practical applications. 

To summarize, our document is organized as follows. First, in Section \ref{sec:mathematical_framework}, we set the mathematical framework of functional regression and recall the definition of the signature and its main properties. Then, in Section \ref{sec:the_expected_sig_model}, we introduce our model, called `signature linear model', and define estimators of its parameters. Their rates of convergence are given in Section \ref{sec:theoretical_results}. Finally, Section \ref{sec:numerical_methodology} is devoted to the practical implementation of the signature linear model. We conclude by demonstrating its performance on both simulated and real-world datasets in Section \ref{sec:experimental_results}.

For the sake of clarity, the proofs of the mathematical results are postponed to Section \ref{sec:proofs}. The code is completely reproducible and available at \url{https://github.com/afermanian/signature-regression}.

\section{Mathematical framework} \label{sec:mathematical_framework}

\subsection{Functional linear regression} \label{sec:functional_linear_regression}

We place ourselves in a functional linear regression setting with scalar responses: we are given a dataset $D_n =  \{(X_1,Y_1),\dots,(X_n,Y_n)\}$, where the pairs $(X_i, Y_i)$ are independent and identically distributed copies of a random couple $(X,Y)$, where $X$ is a (random) function, $X: [0,1] \to \R^d$, $d \geq 1$, and $Y$ a real random variable. Our goal is to approximate the regression function $f(X)=\esp[Y|X]$ by a parametrized linear function $f_\theta$ and to build an estimator of $\theta$.

In the univariate case, that is when $d=1$, the classical functional linear model \citep{frank1993statistical,hastie1993statistical} writes
\begin{equation} \label{eq:functional_linear_model}
Y = \alpha + \int_0^1 X(t) \beta(t)dt + \varepsilon,
\end{equation}
where $\alpha \in \R$, $\beta:[0,1] \to \R$ and $\varepsilon$ is a random noise. The functional coefficients $\beta$ and the functional covariates $X_i$ are then expanded on basis functions:
\begin{equation} \label{eq:basis_expansion}
 \beta(t) = \sum_{k=1}^K b_k \phi_k(t), \qquad X_i (t) = \sum_{k=1}^K c_{ik} \phi_k(t),
 \end{equation}
where $\phi_1, \dots, \phi_K$ are a set of real-valued basis functions (for example the monomials $1, t, t^2, \dots, t^K$ or the Fourier basis). Equation \eqref{eq:functional_linear_model} can then be rewritten in terms of the $c_{ik}$s and $b_k$s, which brings the problem back to the well-known multivariate linear regression setting. Different approaches can then be used in terms of choice of basis functions and regularization \citep[see][Chapter 15]{ramsay2005functional}. Note that another common approach is functional principal components regression \citep{cardot1999functional,brunel2016non}. The idea is to perform a functional principal components analysis (fPCA) on $X$, which gives a representation of $X$ as a sum of $K$ orthonormal principal components, and to use these as basis functions $\phi_k$s.

In both cases, the functional nature of the problem is dealt with by projecting the functions $X$ on a smaller linear space, spanned by basis functions. This basis expansion is not straightforward to extend to the vector-valued case, that is when $d >1$, the common approach being to expand each coordinate of $X$ independently. This amounts to assuming that there are no interactions between coordinates, which is a strong assumption and not an efficient representation when the coordinates are highly correlated. Moreover, to our knowledge, the only theoretical results in the vector-valued case are found in the domain of longitudinal data analysis \citep{greven2011longitudinal,park2015longitudinal}. In this case, the different coordinates are assumed to be repeated measurements of a quantity of interest on a patient and each coordinate is given a parametric model, in the same spirit as ANOVA models. These parametric models do not apply in the general case when the coordinates may correspond to different quantities such as the evolution of different stocks or the $x$-$y$-$z$ coordinates of a pen trajectory.

The signature approach removes the need to make such assumptions: the focus moves from finding a functional model for $X$ to finding a basis for functions of $X$. In other words, instead of using a basis of functions, we use a basis of functions of functions. In a regression setting, this shift of perspective is particularly adequate since the object of interest is the regression function $f(X)$ and not $X$ itself. The whole approach is based on the signature transformation, which takes as input a function $X$ and outputs an infinite vector of coefficients known to characterize $X$ under some smoothness assumptions. In particular, there are no assumptions on the structure of dependence in the different coordinates of $X$. In other words, the signature is naturally adapted to the vector-valued case.

Before we delve into the signature linear model, we gently introduce the notion of signature and review some of its important properties. 


\subsection{The signature of a path}
\label{sec:sig_def_prop}

We give here a brief presentation of signatures but the reader is referred to \citet{lyons2007differential} or \citet{friz2010multidimensional} for a more involved mathematical treatment with proofs. To follow the vocabulary from rough path theory, we will often call the functional covariate $X: [0,1] \to \R^d$ a path. Our basic assumption is that $X$ is of bounded variation, i.e., it has finite length.
\begin{definition}
		Let $X : [0,1] \rightarrow \R^d, \, t\mapsto (X^1_t, \dots , X^d_t)^\top$. The total variation of $X$  is defined by
		\[\|X \|_{TV} =  \underset{{\mathcal I}}{\textnormal{sup}} \sum_{(t_0,\dots,t_k) \in {\mathcal I}} \|X_{t_i} - X_{t_{i-1}} \| ,\]
		where the supremum is taken over all finite subdivisions of $[0,1]$, and $\| \cdot \|$ denotes the Euclidean norm on $\R^d$. The set of paths of bounded variation is then defined by
		\begin{equation*}
			BV(\R^d) = \big\{ X: [0,1] \to \R^d \, | \, \|X \|_{TV} < \infty \big\}.
		\end{equation*}
\end{definition}
We recall that $BV(\R^d)$ endowed with the norm $\|X\|_{BV(\R^d)} = \| X\|_{TV} + \sup_{t \in [0,1]} \| X_t\| $ is a Banach space. We stress that the basis functions traditionnaly used in functional data analysis are of bounded variation. The assumption that $X \in BV(\R^d)$ is therefore much less restrictive than assuming an expansion such as \eqref{eq:basis_expansion}. This assumption allows to define Riemann-Stieljes integrals along paths, which puts us in a position to define the signature. 
	
	\begin{definition}
		Let $X\in BV(\R^d)$ and $I=(i_1,\dots , i_k) \subset \{1,\dots d\}^k$, $k \geq 1$, be a multi-index of length $k$. The signature coefficient of $X$ along the index $I$ on $[0,1]$ is defined by
		\begin{equation}
		\label{eq:def_signature_S^I}
		S^I(X) =  \idotsint\limits_{0 \leq u_1<\dots <u_k \leq 1} dX^{i_1}_{u_1} \dots dX^{i_k}_{u_k}.
		\end{equation}
		$S^I(X)$ is then said to be a signature coefficient of order $k$.
	\end{definition}
	The signature of $X$ is the sequence containing all signature coefficients, i.e.,
	\begin{equation*}
	S(X)=\big(1,S^{(1)}(X),\dots,S^{(d)}(X),S^{(1,1)}(X),S^{(1,2)}(X),\dots, S^{(i_1,\dots, i_k)} (X),\dots  \big) .
	\end{equation*}
	The signature of $X$ truncated at order $m$, denoted by $S^m(X)$, is the sequence containing all signature coefficients of order lower than or equal to $m$, that is
	\begin{equation*}
	S^m(X)=\big(1,S^{(1)}(X),S^{(2)}(X),\dots, S^{\overbrace{(d,\dots ,d)}^{\textnormal{length } m}}(X) \big) .
	\end{equation*}

	Note that the assumption that $X \in BV(\R^d)$ may be relaxed: the signature may still be defined when the Riemann-Stieljes integrals are not well-defined. For example, the signature of the Brownian motion may be defined via Stratonovitch integrals \citep{le2013stratonovich}. Integrating paths that are not of bounded variation is actually one of the motivations behind the definition of the signature in rough path theory.

\colorlet{lightgray}{gray!10}
\definecolor{colorp1}{HTML}{0173b2}
\definecolor{colorp2}{HTML}{de8f05}
\definecolor{colorp3}{HTML}{029e73}

\global\def\plotrange{0, 1, 4, 8, 11, 15, 18, 19}
\def\myarr {
    (0.131579, 0.5)
    (0.2631579, 1.228486295)
    (1.05263158, 2.37243038 - 0.3)
    (2.10526316, 2.232905675 - 0.3)
    (2.89473684 - 0.1, 1.597253615)
    (3.94736842 - 0.3, 1.09415549 + 0.2)
    (4.7368421 - 0.5, 1.6927394599999999)
    (5.0 -0.3, 2.5)
}

\global\def\xs{{0.131579, 0.2631579, 0.52631578, 0.78947368, 1.05263158, 1.31578948, 1.57894736, 1.84210526, 2.10526316, 2.36842106, 2.63157894, 2.89473684 - 0.1, 3.15789474, 3.42105264, 3.68421052, 3.94736842 - 0.3, 4.21052632, 4.47368422, 4.7368421 - 0.5, 5.0 - 0.3}}

\global\def\ys{{0.5, 1.228486295, 1.7697186200000001, 2.14419923, 2.37243038 - 0.3, 2.47491434 - 0.3, 2.47215338, 2.384649725, 2.232905675 - 0.3, 2.037423455, 1.81870535, 1.597253615, 1.39357049, 1.2281582599999998, 1.121519165, 1.09415549 + 0.2, 1.166569475, 1.359263375, 1.6927394599999999, 2.5}}

\global\def\xlen{19}

\begin{figure*}[ht]
    \centering
        \begin{tikzpicture}[scale=1.1, every node/.style={scale=1.1}]
            \draw[xshift=0cm, name path=one] plot coordinates {
                \myarr
            };
            \draw[xshift=0cm, name path=two] plot coordinates {
                (\xs[0], \ys[0]) (\xs[0], \ys[\xlen]) (\xs[\xlen], \ys[\xlen])
            };
            \tikzfillbetween[
                of=one and two, split, on layer=bg
            ] {pattern=north west lines, colorp1!30};
            \draw[xshift=0cm, name path=one] plot coordinates {
                \myarr
            };
            \draw[xshift=0cm, name path=two] plot coordinates {
                (\xs[0], \ys[0]) (\xs[\xlen], \ys[0]) (\xs[\xlen], \ys[\xlen])
            };
            \tikzfillbetween[
                of=one and two, split, on layer=bg
            ] {pattern=north west lines, colorp2!30};
            \draw[step=0.5, draw=black!30] (0, 0) grid (5, 3) rectangle (0, 0);
            \foreach \i in \plotrange {
                \node at (\xs[\i], \ys[\i])[circle, draw=black, fill=colorp3, inner sep=1.5pt] {};
            }
            \draw[black, ->] (0, 0) -- (0, 3.03);
            \draw[black, ->] (0, 0) -- (5.03, 0);
            \node at (5.3, 0) {\small $X^i$};
            \node at (0, 3.3) {\small $X^j$};
            \pgfmathsetmacro{\xprev}{\xs[0] + 0.05}
            \pgfmathsetmacro{\yprev}{\ys[0]}
            \foreach \i in \plotrange {
                \pgfmathsetmacro{\x}{\xs[\i]}
                \pgfmathsetmacro{\y}{\ys[\i]}
                \node at (\x, \y)[circle, draw=black, fill=colorp3, inner sep=1.5pt] {};
                \draw[colorp3, thick, cap=round] (\xprev, \yprev) -- (\x, \y);
                \global\let\yprev=\y
                \global\let\xprev=\x
            }
            \node at (2.5, 2.0) [right=0.1mm] {\footnotesize $S^{(i, j)}(X)$};
            \node at (1.2, 1.1) [right=0.1mm] {\footnotesize $S^{(j, i)}(X)$};
            
            \draw[black!60, line width=0.3mm, dashed, ->] (\xs[0] + 0.05, \ys[0]) -- (\xs[\xlen] - 0.01, \ys[0]);
            
            \draw[black!60, line width=0.3mm, dashed, ->] (\xs[\xlen], \ys[0]) -- (\xs[\xlen], \ys[\xlen] - 0.1);
            
            \node at (2.5, \ys[0]) [below] {\footnotesize $S^{(i)}(X)$};
            \node at (\xs[\xlen], 1.5)  [right] {\footnotesize $S^{(j)}(X)$};

        \end{tikzpicture}
    
    \caption{Geometric interpretation of the signature coefficients. The terms $S^{(i)}(X)$ and $S^{(j)}(X)$ are the increments of the coordinates $i$ and $j$ respectively. The terms  $S^{(i, j)}$ and $S^{(j, i)}$ correspond to the areas of the blue and orange regions respectively.}
    \label{fig:geometric_signature}
\end{figure*}
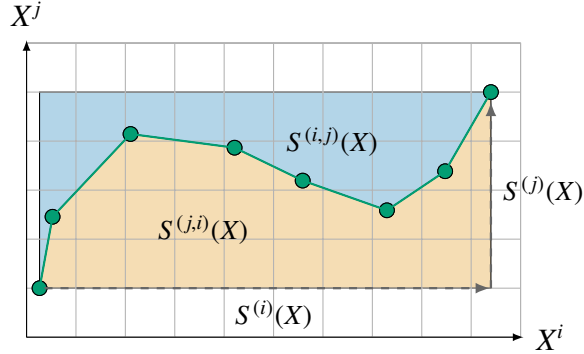

	A crucial feature of the signature is that it encodes the geometric properties of the path, as shown in Fig.~\ref{fig:geometric_signature}. Indeed, coefficients of order 1 correspond to the increments of the path in each coordinate and the coefficients of order 2 correspond to areas outlined by the path. For higher orders of truncation, the signature contains information about the joint evolution of tuples of coordinates. Moreover, it is clear from its definition as an integral that the signature is independent of the time parametrization \citep[][Proposition 7.10]{friz2010multidimensional} and that it is invariant by translation. Therefore, the signature looks at functions as purely geometric objects, without any information about sampling frequency, speed, or travel time, hence the terminology of `paths'. 

	Note that the definition can be extended to paths defined on any interval $[s,t] \subset \R$ by changing the integration bounds in \eqref{eq:def_signature_S^I}. Moreover, it is clear that there are $d^k$ signature coefficients of order $k$. The signature truncated at order $m$ is therefore a vector of dimension $s_d(m)$, where
	\[s_d(m)=\sum_{k=0}^{m} d^k = \frac{d^{m+1}-1}{d-1} \quad \text{if} \quad  d \geq 2, \quad s_d(m)=m+1 \quad \text{if} \quad d=1.\]
	Thus, provided $d \geq 2$, the size of $S^m(X)$ increases exponentially with $m$ and polynomially with $d$---some typical values are presented in Table \ref{tab:truncation_dimensions}. 
	
	\begin{table}[ht]
		\centering
		\caption{Typical values of $s_d(m)$, the size of the signature of a path $X \in BV(\R^d)$ truncated at order $m$.}
		\label{tab:truncation_dimensions}
		\begin{tabular}{lccc}
			& $d=2$  & $ d=3$   &$ d=6$ \\
			\hline
			$m=1$& 2 & 3 & 6 \\
			$m=2$& 6 & 12 & 42 \\
			$m=5$ & 62 &363 & 9330 \\
			$m=7$ &254 & 3279 & 335922\\
			\hline
		\end{tabular}
	\end{table}
	The set of coefficients of order $k$ can be seen as an element of the $k$th tensor product of $\R^d$ with itself, denoted by $(\R^d)^{\otimes k}$. For example, the $d$ coefficients of order $1$ can be written as a vector, and the $d^2$ coefficients of order $2$ as a matrix, i.e.,
	\[\begin{pmatrix}S^{(1)} (X) \\ \vdots \\ S^{(d)} (X) \\ \end{pmatrix} \in \R^d,   \quad \begin{pmatrix}
	S^{(1,1)} (X)& \dots & S^{(1,d)} (X)\\
	\vdots  & \ddots & \vdots \\
	S^{(d,1)}(X) & \dots & S^{(d,d)} (X)\\
	\end{pmatrix} \in \R^{d \times d} \approx (\R^d)^{\otimes 2}.\]
	Similarly, coefficients of order 3 can be written as a tensor of order 3, and so on. Then, $S(X)$ can be seen as an element of the tensor algebra
	\[\R \oplus \R^d \oplus (\R^d)^{\otimes 2} \oplus \dots \oplus (\R^d)^{\otimes k} \oplus \cdots. \]
	Although not fundamental in the present paper, this structure of tensor algebra is the right space to understand properties of the signature \citep{lyons2007differential,friz2010multidimensional}.
	
	Let us give two examples of paths and their signatures.

	\begin{example}
	Let $X$ be a parametrized curve: for any $t \in [0,1]$, $X_t = (t, f(t))$, where $f: \R \to \R$ is a smooth function. Then,
	\begin{align*}
		S^{(1)}(X) = \int_0^1dX^1_t = \int_0^1dt = 1,\quad  S^{(2)}(X) = \int_0^1dX^2_t = \int_0^1f'(t)dt = f(1) - f(0),
	\end{align*}
	where $f'$ denotes the derivative of $f$. Similarly, the signature coefficient along $(1,2)$ is
		\begin{align*}
		S^{(1,2)}(X) &= \int_0^1 \int_0^t dX^1_u dX^2_t = \int_0^1 \Big(\int_0^t du \Big) f'(t)dt = \int_0^1 t f'(t)dt = f(1) - \int_0^1 f(t) dt.
	\end{align*}
	\end{example}

	\begin{example}
	Let $X$ be a $d$-dimensional linear path: \[X_t= \begin{pmatrix}
	X^1_t \\ \vdots \\ X^d_t
	\end{pmatrix} = \begin{pmatrix}
	a_1 + b_1t \\ \vdots \\ a_d + b_dt
	\end{pmatrix}.\] Then, for any index $I=(i_1,\dots,i_k) \subset \{1,\dots, d\}^k$, the signature coefficient along $I$ is
	\begin{equation}
	\label{eq:formula_sig_linear_path}
	S^{(i_1,\dots ,i_k)}(X)=  \idotsint\limits_{0 \leq u_1<\dots <u_k \leq 1} dX^{i_1}_{u_1} \cdots dX^{i_k}_{u_k} = \idotsint\limits_{0 \leq u_1<\dots <u_k \leq 1} b_{i_1} du_1 \cdots b_{i_k} du_k= \frac{b_{i_1} \dots b_{i_k}}{k!}.
	\end{equation}
	It is clear here that the signature is invariant by translation: $S(X)$ depends only on the slope of $X$ and not on the initial position $(a_1, \dots, a_d)^\top \in \R^d$.
	\end{example}

	We now recall a series of properties of the signature that motivate the definition of the signature linear model. The first important property provides a criterion for the uniqueness of signatures.
    \begin{proposition}
        \label{prop:uniqueness}
        Assume that $X \in BV(\R^d)$ contains at least one monotone coordinate, then $S(X)$ characterizes $X$ up to translations and reparametrizations.
    \end{proposition}
    This is a sufficient condition, a necessary one has been derived by \citet{hambly2010uniqueness} and is based on the construction of an equivalence relation between paths, called tree-like equivalence. For any path $X \in BV(\R^d)$, the time-augmented path $\tilde{X}_t = (X_t,t)^\top \in BV(\R^{d+1})$ satisfies the assumption of Proposition \ref{prop:uniqueness}, which ensures signature uniqueness. Enriching the path with new dimensions is actually a classic part of the learning process when signatures are used, and is discussed by \citet{fermanian2021embedding} and \citet{morrill2020generalised}. We will always use this time-augmentation transformation before computing signatures.

    The next proposition states that the signature linearizes functions of $X$ and is the core motivation of the signature linear model. We refer the reader to \citet[]{levin2013learning}, Theorem 3.1, for a proof in a similar setting.
    \begin{proposition}
        \label{prop:linear_approx}
        Let $D \subset BV(\R^d)$ be a compact set of paths that such that, for any $X \in D$, $X_0=0$, and denote by $\tilde{X}= (X_t, t)^\top_{t \in [0,1]}$ the associated time-augmented path. Let $f : D \rightarrow \R$ be a continuous function. Then, for every $\varepsilon >0$, there exist $m^\ast \in \N$, $\beta^\ast \in \R^{s_d(m^\ast)}$, such that, for any $X\in D$,
        \[\big|f(X) - \langle \beta^\ast, S^{m^\ast}(\tilde{X}) \rangle \big| \leq \varepsilon,\]
        where $\langle \cdot,\cdot \rangle$ denotes the Euclidean scalar product on $\R^{s_d(m^\ast)}$.
    \end{proposition}
    This proposition is a consequence of the Stone-Weierstrass theorem. The classical Weierstrass approximation theorem states that every real-valued continuous function on a closed interval can be uniformly approximated by a polynomial function. Linear forms on the signature can, therefore, be thought of as the equivalent of polynomial functions for paths. The assumption that $X_0 = 0$ is due to the fact that signatures are invariant by translation: no information about the initial position of the path is contained in signatures.

    Finally, the following bound on the norm of the truncated signature allows to control the rate of decay of signature coefficients of high order---see \citet[][Lemma 5.1]{lyons2014rough} for a proof.

    \begin{proposition}
     \label{prop:exp_decay}
     Let $X:[0,1] \to \R^d$ be a path in $BV(\R^d)$. Then, for any $m \geq 0$, 
       \[\| S^m(X) \| \leq \sum_{k=0}^{m} \frac{ \| X\|_{TV}^k}{k!} \leq e^{\| X\|_{TV}}.\]
    \end{proposition}

\section{The signature linear model}
\label{sec:the_expected_sig_model}

\subsection{Presentation of the model}
\label{sec:def_model}

We are now in a position to present the signature linear model. Recall that our goal is to model the relationship between a real random variable $Y \in \R$ and a random input path $X \in BV(\R^d)$. Without loss of generality, we now assume that $d \geq 2$ and that $X$ has been augmented with time---in other words, one coordinate of $X$ is $t \mapsto t$. Proposition \ref{prop:linear_approx} then motivates the following model which was first introduced in a slightly different form by \cite{levin2013learning}: we assume that there exists $m\in \N$, $\beta_m^\ast \in \R^{s_d(m)}$, such that
\begin{equation}
\label{eq:model_def}
 \esp \left[Y | X \right] = \big\langle \beta_m^\ast, S^{m}(X) \big\rangle, \quad \text{Var}(Y | X) \leq \sigma^2 < \infty.
\end{equation}
We consider throughout the article the smallest $m^\ast \in \N$ such that there exists $\beta^\ast_{m^\ast} \in \R^{s_d(m^\ast)}$ satisfying 
\[\esp \left[Y | X \right] = \big\langle \beta_{m^\ast}^\ast, S^{m^\ast}(X) \big\rangle.\] 
In other words, we assume a regression model, where the regression function is a linear form on the signature. A few comments are in order. 

From an approximation point of view, this model is very general. Indeed, by Proposition \ref{prop:linear_approx}, the only requirements for model \eqref{eq:model_def} to be valid are the continuity of the regression function $f(X) = \esp[Y |X]$ and the fact that $S(X)$ must characterize the random path $X$. The latter is ensured by using a time augmentation, that is, considering $\tilde{X}_t= (X_t, t)$, and by fixing the initial value, for example $X_0=0$. Then, under the assumption that the data is in a compact set---which will be guaranteed later on by assumption $(H_K)$---, for any threshold $\varepsilon >0$, there exist $m^\ast \in \N$ and $\beta^\ast_{m^\ast} \in \R^{s_d(m^\ast)}$ such that
\[\big| \esp[Y|X] - \langle \beta^\ast_{m^\ast}, S^{m^\ast}(X) \rangle \big| \leq \varepsilon.\]
In other words, we know that (the first part of) model \eqref{eq:model_def} is true up to an error of $\varepsilon$. A striking fact is that no assumption that $\esp[Y|X]$ is linear in $X$ is needed, contrary to functional models of the form \eqref{eq:functional_linear_model}.

It is instructive to further compare this model to the functional model \eqref{eq:functional_linear_model}. Much fewer assumptions on $X$ are needed: it is only assumed to be of finite variation, whereas in \eqref{eq:functional_linear_model} it has to have a finite basis expansion. Moreover, our model is directly adapted to the vector-valued case. Finally, it depends directly on a finite vector $\beta^\ast_{m^\ast}$, whereas $\eqref{eq:functional_linear_model}$ is written in terms of a function $\beta$, which must itself be written on basis functions. Note that the choice of basis needs to be adapted to each particular application, whereas the signature linear model only depends on two parameters. In a nutshell, it is a more general model with fewer hyperparameters.

It can be noticed that, since the first term of signatures is always equal to 1, this regression model contains an intercept: when $m^\ast=0$, \eqref{eq:model_def} is a constant model. There are two unknown quantities in model \eqref{eq:model_def}: $m^\ast$ and $\beta^\ast_{m^\ast}$. The parameter $m^\ast$ is the truncation order of the signature of $X$ and controls the model size, whereas $\beta^\ast_{m^\ast}$ is the vector of regression coefficients, whose size $s_d(m^\ast)$ depends on $m^\ast$.


The signature truncation order $m^\ast$ is a key quantity in this model and influences the rest of the study. Indeed, it controls the number of coefficients and therefore the computational feasibility of the whole method. However, it is in general little discussed in the literature and small values are picked arbitrarily, regardless of the model used on top of signatures. For example, \cite{liu2017ps} consider values of $m$ up to 2, \cite{yang2015chinese} up to 3, \citet{arribas2018signature} and \citet{lai2017online} up to 4, \cite{yang2016deepwriterid} up to 5 , and \cite{yang2017leveraging} up to 8. Thus, one of our main objectives is to establish a rigorous procedure to estimate $m^\ast$, and, to this end, we define a consistent estimator of $m^\ast$. As we will see later, a simple estimator of $\beta^\ast_{m^\ast}$, and therefore of the regression function, is then also obtained.

\subsection{Estimating the truncation order}
\label{sec:def_estimators}

Let $D_n =  \{(X_1,Y_1),\dots,(X_n,Y_n)\}$ be a set of i.i.d.~observations drawn according to the law of $(X,Y)$. We use the approach of penalized empirical risk minimization. For the moment, let us fix a certain truncation order $m \in \N$, and let $\alpha>0$ denote a fixed positive number. Then, the ball in $\R^{s_d(m)}$ of radius $\alpha$ centered at $0$ is denoted by
\[B_{m,\alpha} =\big\{ \beta \in \R^{s_d(m)} \, | \, \left\| \beta \right\| \leq \alpha \big\},\]
where $\| \cdot \|$ stands for the Euclidean norm, whatever the dimension. By a slight abuse of notation,  the sequence $( B_{m,\alpha})_{m \in \N}$ can be seen as a nested sequence of balls, i.e., $B_{0,\alpha} \subset B_{1,\alpha} \subset \dots \subset B_{m,\alpha} \subset B_{m+1,\alpha} \subset \cdots .$ From now on, we will only consider coefficients within these balls. Therefore, we assume that the true coefficient $\beta^\ast_{m^\ast}$ lies within such a ball, i.e., we make the assumption
\begin{itemize}
\item[$(H_{\alpha})$] There exists $\alpha >0$ such that $\beta^\ast_{m^\ast} \in B_{m^\ast,\alpha}. $
\end{itemize}
On the one hand, for a fixed truncation order $m$, the theoretical risk is defined by ${\mathcal{R}}_{m}(\beta ) = \esp \big(Y - \big\langle \beta, S^m(X) \big\rangle \big)^2.$ Then, the minimal theoretical risk for a certain truncation order $m$, is defined by 
\[L(m) = \underset{\beta \in B_{m,\alpha}}{\inf} {\mathcal{R}}_{m}(\beta) = {\mathcal{R}}_{m}(\beta^\ast_{m}),\]
where $\beta^\ast_{m} \in  \text{argmin }_{\beta \in B_{m,\alpha}} \risk_{m}(\beta)$ (note that the existence of $\beta^\ast_m$ is ensured by convexity of the problem). Since the sets $(B_{m,\alpha})_{m \in \N}$ are nested, $L$ is a decreasing function of $m$. Its minimum is attained at $m=m^\ast$, and, provided $m \geq m^\ast$, $L(m)$ is then constant and equal to 
\[\risk(\beta^\ast_{m^\ast}) = \esp \big(Y - \big\langle \beta^\ast_{m^\ast}, S^{m^\ast}(X) \big\rangle \big)^2 = \esp \big( \textnormal{Var}(Y | X) \big) \leq \sigma^2.\] 
On the other hand, the empirical risk with signature truncated at order $m$ is defined by $\widehat{\risk}_{m,n}(\beta) = \frac{1}{n}\sum_{i=1}^{n} \big(Y_i -\big\langle \beta, S^m(X_i)\big\rangle \big)^2,$ where $\beta \in B_{m,\alpha}$. The minimum of $\widehat{\risk}_{m,n}$ over $B_{m,\alpha}$ is denoted by $\widehat{L}_n(m)$ and defined as
\begin{equation*}
\widehat{L}_n(m) = \underset{\beta \in B_{m,\alpha}}{\min} \widehat{\risk}_{m,n}(\beta)=\widehat{\risk}_{m,n}(\widehat{\beta}_m),
\end{equation*}
where $\widehat{\beta}_m$ denotes a point in $B_{m,\alpha}$ where the minimum is attained. Note that $\beta \mapsto \widehat{\risk}_{m,n}(\beta)$ is a convex function so $\widehat{\beta}_m$ exists. We point out that minimizing $\widehat{\risk}_{m,n}$ over $B_{m,\alpha}$ is equivalent to performing a Ridge regression with a certain regularization parameter which depends on $\alpha$. 

To summarize, for a fixed truncation order $m$, a Ridge regression gives the best parameter $\widehat{\beta}_m$ to model $Y$ as a linear form on the signature of $X$ truncated at order $m$. Recall that our goal is to find a truncation order $\widehat{m}$ close to the true one $m^\ast$. Since the $(B_{m,\alpha})_{m \in \N}$ are nested, the sequence $(\widehat{L}_n(m))_{m \in \N}$ decreases with $m$. Indeed, increasing $m$ makes the set of parameters larger and therefore decreases the empirical risk. An estimator of $m^\ast$ can then be defined by a trade-off between this decreasing empirical risk and an increasing function that penalizes the number of coefficients:
 \begin{equation*}
\label{eq:def_hatm}
\widehat{m} = \text{min} \Big( \underset{m \in \N}{\text{argmin}} \big(\widehat{L}_n(m) + \text{pen}_n(m)\big) \Big),
\end{equation*}
where $m \mapsto \text{pen}_n(m)$ is an increasing function of $m$ that will be defined in Theorem \ref{th:main_bound_hatm}. If the minimum is reached by several values, we set $\widehat{m}$ to the smallest one. The procedure is illustrated in Fig.~\ref{fig:toy_example_pen_curve} for a toy dataset which will be described in Section \ref{sec:toy_example}.

\begin{figure}[ht]
    \centering
    \includegraphics[width=.5\textwidth]{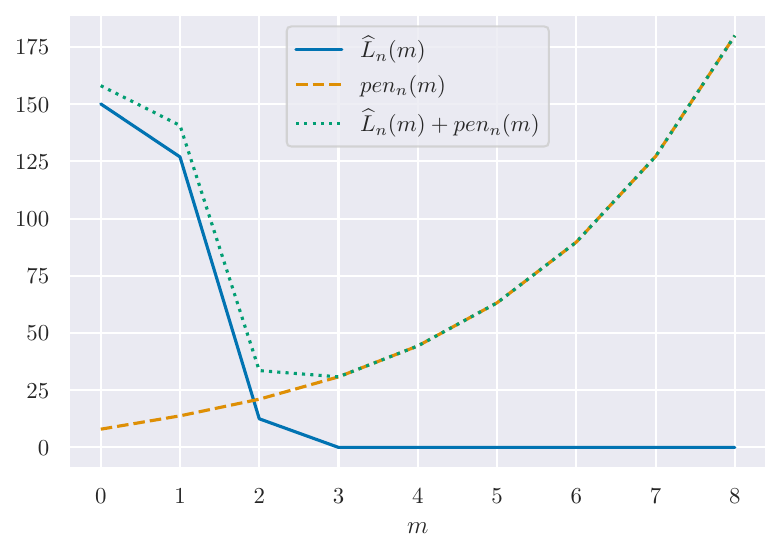} 
    \caption{The functions $m \mapsto \widehat{L}_n(m)$ (blue solid curve), $m \mapsto \text{pen}_n(m)$ (orange dashed curve) and $m \mapsto \widehat{L}_n(m) + \text{pen}_n(m)$ (green dotted curve) for a toy dataset. In this case, the value of $\widehat{m}$ is $\widehat{m}=3$.}
    \label{fig:toy_example_pen_curve}
\end{figure}

Now that we have an estimate of $m^\ast$, which is a key ingredient in establishing the whole process of the expected signature method, and before presenting the whole procedure, we justify the estimator by some theoretical results.

\section{Performance bounds}
\label{sec:theoretical_results}

In this section, we show that it is possible to calibrate a penalization that ensures exponential convergence of $\widehat{m}$ to $m^\ast$. In addition to $(H_{\alpha})$, we need the following assumption:
\begin{itemize}
	\item[$(H_K)$] there exists $K_Y > 0$ and $K_X > 0$ such that almost surely $|Y| \leq K_Y$ and $\|X \|_{TV} \leq K_X$.
\end{itemize}
The assumption $(H_K)$ says that the trajectories have a length uniformly bounded by $K_X$ and that the responses $Y$ live in a compact set. These assumptions are quite different from the ones in functional linear models of the form \eqref{eq:functional_linear_model}. Indeed, concerning the regularity of $X$, they typically assume that $X$ is in $L^2$ and that its coefficients $c_{ik}$ in the basis expansion \eqref{eq:basis_expansion} decrease sufficiently fast. We therefore trade an assumption that the functions have a nice basis decomposition for a compactness property, which seems a reasonable choice for practical applications. For example, any discrete-time time-series model observed over a finite horizon, such as ARIMA, satisfies $(H_K)$. Any continuously differentiable function with bounded derivative also satisfies $(H_K)$. Note also that $(H_K)$ does not depend strongly on the dimension $d$, whereas the assumptions of functional linear models become very stringent in this case; they typically assume an additive relationship between $Y$ and the different coordinates of $X$. We shall also use the constant $K$, defined by
\begin{equation}
\label{eq:def_K}
K=2 (K_Y + \alpha  e^{K_X}) e^{K_X}.
\end{equation}
The main result of the section is the following.

\begin{theorem}
	\label{th:main_bound_hatm}
	Let $K_{\textnormal{pen}} > 0$, $0 < \rho < \frac{1}{2}$, and
	\begin{equation}
	\label{eq:def_pen}
	\textnormal{pen}_n(m) = K_{\textnormal{pen}} n^{-\rho} \sqrt{s_d(m)}.
	\end{equation}
	Let $n_0$ be the smallest integer satisfying
	\begin{align}
	\label{eq:condition_n_0_thm}
	(n_0)^{\tilde{\rho}} \geq  & (432 K \alpha \sqrt{\pi} + K_{\textnormal{pen}})\Big( \frac{2\sqrt{s_d(m^\ast+1)}}{L(m^\ast -1) - \sigma^2}
	 + \frac{\sqrt{2s_d(m^\ast+1)}}{K_{\textnormal{pen}}\sqrt{d^{m^\ast+1}}} \Big) ,
	\end{align}
	where $\tilde{\rho}=\min(\rho, \frac{1}{2} - \rho)$. Then, under the assumptions $(H_{\alpha})$ and $(H_K)$, for any $n \geq n_0$,
	\begin{equation*}
	\label{eq:prob_hatm_neq_mast}
	\prob \left( \widehat{m} \neq m^\ast\right) \leq C_1 \exp \left(- C_2n^{1-2\rho}\right),
	\end{equation*}
	where the constants $C_1$ and $C_2$ are defined by 
	\begin{equation}
	\label{eq:def_C_1}
	C_1 =  74 \sum_{m>0} e^{-C_3s_d(m)} + 148 m^\ast, \quad C_3 =\frac{K_{\textnormal{pen}}^2 d^{m^\ast+1}}{128s_d(m^\ast+1)(72K^2\alpha^2 + K_Y^2)},
	\end{equation}
	and
	\begin{equation}
	\label{eq:def_C_2}
		C_2 = \frac{1}{16(1152 K^2 \alpha^2 + K_Y^2)} \min \Big(\frac{K_{\textnormal{pen}}^2  d^{m^\ast+1}}{8s_d(m^\ast+1)}, L(m^\ast-1) - \sigma^2\Big).
	\end{equation}
\end{theorem}
This theorem provides a non-asymptotic bound on the convergence of $\widehat{m}$. It implies the almost sure convergence of $\widehat{m}$ to $m^\ast$. We can note that the penalty decreases slowly with $n$ (more slowly than a square-root) and increases with $m$ exponentially, i.e.,  as $d^{m/2}$. The penalty includes an arbitrary constant $K_{\textnormal{pen}}$. Its value that minimizes $n_0$ is 
\[K_{\textnormal{pen}}^\ast = \sqrt{ \frac{(L(m^\ast-1)-\sigma^2) 432 \sqrt{\pi} \alpha K}{d^{m^\ast+1}}},\]
and, in practice, it is calibrated with the slope heuristics method of \cite{birge2007minimal}, described in Section \ref{sec:numerical_methodology}. The proof of Theorem \ref{th:main_bound_hatm} is based on chaining tail inequalities that bound uniformly the tails of the risk. We refer the reader to Section \ref{sec:proofs} for a detailed proof. 

To give some insights into this estimator it is interesting to look at the behavior of the constants when different quantities vary.
\begin{itemize}
\item If the dimension of the path $d$ gets large, then $d^{m^\ast +1} \sim s_d(m^\ast+1)$ and the constants $C_1$ and $C_2$ stay of the same order (provided that the risk $L(m^\ast-1)$ stays constant). Therefore, the quality of the bound does not change in high dimensions. However, the constant $n_0$ increases at the rate of $\mathcal{O}(d^{m^\ast/2\tilde{\rho}})$: we neeed exponentially more data when $d$ grows.
\item If the true truncation parameter $m^\ast$ is large, that is, the regression function $\esp[Y|X]$ depends on higher-order terms of the signature, the same phenomenon is observed except that $C_1$ increases linearly: $C_2$ and $C_3$ stay of the same order, $C_1 \sim 148m^\ast$, and $n_0$ increases at the rate of $\mathcal{O}(d^{m^\ast/2\tilde{\rho}})$. It is not surprising: when $m^\ast$ increases, the size of the coefficient $\beta^\ast_{m^\ast}$ increases and therefore more data are needed to estimate it. 
\item If $\alpha$ increases, $n_0$ and $C_1$ increase while $C_2$ decreases. In other words, more data is needed and the quality of the estimator deteriorates. Indeed, when $\alpha$ gets larger, the parameter spaces $B_{m,\alpha}$ gets larger for any $m$ so estimation is harder.
\item The last quantity of interest is $L(m^\ast-1) - \sigma^2 \leq L(m^\ast-1) - L(m^\ast)$, which measures the difference of risk between a smaller model and the model truncated at $m^\ast$. By definition, it is a strictly positive quantity. When it gets close to zero, it means that a model truncated at $m^\ast-1$ is almost as good as a model truncated at $m^\ast$. We can see that when this difference decreases, $n_0$ increases and $C_2$ decreases: it is harder to find that a truncation order of $m^\ast$ is better than $m^\ast-1$, therefore the estimator $\widehat{m}$ deteriorates.
\end{itemize}

With an estimator of $\widehat{m}$ at hand, one can simply choose to estimate $\beta^\ast_{m^\ast}$ by $\widehat{\beta}_{\widehat{m}}$, which gives an estimator of the regression function in model \eqref{eq:model_def}. As a by-product of Theorem \ref{th:main_bound_hatm}, we then get the following bound.

\begin{corollary}
\label{cor:estimator_hat_beta}
Under the assumptions $(H_{\alpha})$ and $(H_K)$, for any $n \geq n_0$,
\begin{equation*}
\label{eq:l2_cvg_reg_function}
\esp \Big( \big\langle \hat{\beta}_{\widehat{m}},S^{\widehat{m}}(X) \big\rangle - \big\langle \beta^\ast_{m^\ast},S^{m^\ast}(X) \big\rangle  \Big)^2 \leq \frac{C_5}{\sqrt{n}} + C_6 e^{-C_2 n^{1-2\rho}},
\end{equation*}
where the constants $C_5$ and $C_6$ are defined by
\begin{equation*}
	C_5 = 36 K \alpha \sqrt{\pi}(m^\ast +1) \sqrt{s_d(m^\ast)}, \quad C_6 = 2664 K \alpha \sqrt{\pi} \sum_{m>m^\ast} \sqrt{s_d(m)} e^{-C_3 s_d(m)}+ 2 \alpha^2 e^{K_X}  C_1.
\end{equation*}
\end{corollary}
The proof is given in Section \ref{sec:proofs}. This rate of convergence in $\mathcal{O}(n^{-1/2})$ is similar to the ones usually obtained for functional linear models when $d=1$, except that much less assumptions are needed on the path $X$. Indeed, the rates obtained on the regression function usually depend on regularity assumptions on $X$ and $\beta$ in \eqref{eq:functional_linear_model}. For example, it can depend on the Fourier coefficients of $X$ \citep{hall2007methodology}, on the number of Lipschitz-continuous derivatives of $\beta$ \citep{cardot2003spline}, or on the periodicity of $X$ \citep{li2007rates}. We can note that when the true coefficient $m^\ast$ gets larger, prediction is more difficult and the bound increases. This is also the case when $K$ increases, which amounts to allowing larger values for $Y$ and $X$.

We stress that in both Theorem \ref{th:main_bound_hatm} and Corollary \ref{cor:estimator_hat_beta}, the constant $\alpha$ is assumed to be fixed. In practice, it is unknown and is typically selected via cross-validation. Taking this into account in the theoretical analysis would be an interesting extension for future work. We have now all the ingredients necessary to implement this signature linear model. Before looking at its performance on real-world datasets, we present in the next section the complete methodology from a computational point of view.

\section{Computational aspects}
\label{sec:numerical_methodology}

\subsection{The signature linear model algorithm}
\label{sec:computing_signature}

The first step towards practical application is to be able to compute signatures efficiently. Typically, the input data consists of arrays of sampled values of $X$. We choose to interpolate the sampled points linearly, and therefore our problem reduces to computing signatures of piecewise linear paths. To this end, equation \eqref{eq:formula_sig_linear_path} gives the signature of a linear path and Chen's theorem \citep{chen1958integration}, stated below, provides a formula to compute recursively the signature of a concatenation of paths. 

Let $X:[s,t] \rightarrow \R^d$ and $Y : [t,u] \rightarrow \R^d$ be two paths, $0 \leq s < t < u \leq 1$. The concatenation of $X$ and $Y$, denoted by $X \ast Y$, is defined as the path from $[s,u]$ to $\R^d$ such that, for any $v \in [s,u]$,
\[(X \ast Y)_v = \begin{cases}
    X_v, &\textnormal{ if } v \in [s,t], \\
    X_t + Y_{v} - Y_t, & \textnormal{ if } v \in [t,u].\\
    \end{cases} 
    \]
    
    \begin{proposition}[Chen]
        \label{prop:chen}
        Let $X:[s,t] \rightarrow \R^d$ and $Y : [t,u] \rightarrow \R^d$  be two paths with bounded variation. Then, for any multi-index $(i_1,\dots ,i_k) \subset \{1,\dots ,d\}^k$,
        \begin{equation}
        \label{eq:chen_identity}
        S^{(i_1,\dots,i_k)}(X \ast Y) =  \sum_{\ell=0}^{k} S^{(i_1,\dots, i_\ell)}(X) \cdot S^{(i_{\ell+1}, \dots, i_k)}(Y).
        \end{equation}
    \end{proposition}
    
    This proposition is an immediate consequence of the linearity property of integrals \citep[Theorem 2.9]{lyons2007differential}. Therefore, to compute the signature of a piecewise linear path, it is sufficient to iterate the following two steps:
    \begin{enumerate}
        \item Compute with \eqref{eq:formula_sig_linear_path} the signature of a linear section of the path;
        \item Concatenate it to the other pieces with Chen's formula \eqref{eq:chen_identity}.
    \end{enumerate}
    
    This procedure is implemented in the Python library \pkg{iisignature}   \citep{reizenstein2018iisignature}. Thus, for a sample consisting of $p$ points in $\R^d$, if we consider the path formed by their linear interpolation, the computation of the path signature truncated at level $m$ takes $\mathcal{O}(pd^m)$ operations. The complexity is therefore linear in the number of sampled points but exponential in the truncation order $m$, which emphasizes once more the importance of the choice of $\widehat{m}$.

    \begin{algorithm}[ht]
\caption{Pseudo-code for the signature linear model.}
\label{algo:procedure}
\KwData{$\{ (\mathbf{x}_1,Y_1), \dots, (\mathbf{x}_n,Y_n) \}$}
\KwResult{Estimators $\hat{m}$ and $\hat{\beta}_{\hat{m}}$}

Interpolate linearly the columns of $\mathbf{x}_i$ so as to have a set of continuous piecewise linear paths $X_i:[0,1] \to \R^d$, $1 \leq i \leq n$. Add a time dimension, i.e., consider the path $\widetilde{X}_{i}: [0,1] \to \R^{d+1}$, where $\widetilde{X}^j_i=X^j_i$ for $1 \leq j \leq d$, and $X^{d+1}_{i,t}=t$, $t \in [0,1]$.

Select the Ridge regularization parameter $\lambda$ by cross validation on the regression model with $\big\{ S^1(\widetilde{X}_1),\dots,S^1(\widetilde{X}_n) \big\}$ as predictors.

\For{$m=1,\dots, M$}{

	Compute signatures truncated at level $m$: $\big\{S^m(\widetilde{X}_1),\dots, S^m(\widetilde{X}_n) \big\}$.

	Fit a Ridge regression on the pairs $\big\{ (S^m(\widetilde{X}_1),Y_1),\dots, (S^m(\widetilde{X}_n,Y_n) \big\}$. Compute its squared loss $\widehat{L}_n(m)$. 
       
	Compute the penalization $\textnormal{pen}_n(m)=K_{\textnormal pen} \frac{\sqrt{s_d(m)}}{n^\rho}$. 

}

Choose $\widehat{m}=\underset{0 \leq m \leq M}{\textnormal{argmin }} \big( \widehat{L}_n(m) + \textnormal{pen}_n(m) \big)$. 

Compute $\widehat{\beta}_{\widehat{m}}$ by fitting a Ridge regression on $\big\{ (S^{\widehat{m}}(\widetilde{X}_1),Y_1),\dots, (S^{\widehat{m}}(\widetilde{X}_n,Y_n) \big\}$: $\widehat{\beta}_{\widehat{m}} = (\mathbf{S}^\top \mathbf{S} + \lambda \mathbf{I})^{-1} \mathbf{S}^\top \mathbf{Y},$
where $\mathbf{S} \in \R^{n \times s_d(\widehat{m})}$ is the matrix which rows are the signatures of the inputs $S^{\widehat{m}}(\widetilde{X}_i)^\top$, $\mathbf{I} \in \R^{s_d(\widehat{m}) \times s_d(\widehat{m})}$ is the identity matrix, and $\mathbf{Y} = (Y_1, \dots, Y_n)^\top \in \R^d$ is the vector of responses.

\end{algorithm}

In practice, we are given a dataset $\{ (\mathbf{x}_1,Y_1), \dots, (\mathbf{x}_n,Y_n) \}$, where, for any $1 \leq i \leq n$, $Y_i \in \R$ and $\mathbf{x}_i \in \R^{d\times p_i}$. The columns of the matrix $\mathbf{x}_i$ correspond to values of a process $X_i$ in $\R^d$ sampled at $p_i$ different times. We fix $M \in \N$ such that, for any $m \geq M$, the function $m \mapsto  \widehat{L}_n(m) + \textnormal{pen}_n(m)$ is strictly increasing and apply the procedure described in Algorithm \ref{algo:procedure}.

Note that in the first step of Algorithm \ref{algo:procedure} there exist other choices for the embedding of the matrix $\mathbf{x_i}$ into a continuous path $\widetilde{X}_i$ \citep{fermanian2021embedding}. The parameter $\rho$ is set to 0.4. The constant $K_{\textnormal{pen}}$ is calibrated with the so-called slope heuristics method, first proposed by \cite{birge2007minimal}. 

\subsection{A toy example}\label{sec:toy_example}

This section is devoted to illustrating the different steps of Algorithm \ref{algo:procedure} and the convergence of the estimator $\widehat{m}$ with simulated data. We first simulate a dataset $\{(\mathbf{x_1},Y_1), \dots, (\mathbf{x_n}, Y_n)\}$ following the signature model \eqref{eq:model_def}. 

\begin{figure}[ht]
\centering
    \begin{minipage}{.4\textwidth}
    	    \centering
    \includegraphics[width=\textwidth]{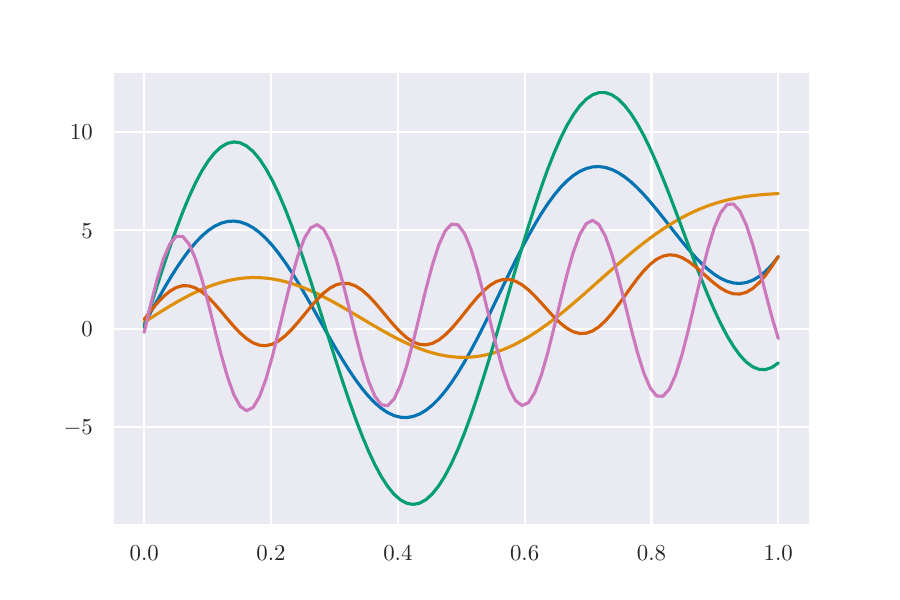} 
    \caption{One sample $X_i$ from model \eqref{eq:simulations_polysinus} with $d=5$.}
    \label{fig:samples_X_polysinus_independent}
    \end{minipage}
    	\begin{minipage}{.4\textwidth}
	    \centering
	    \includegraphics[width=\textwidth]{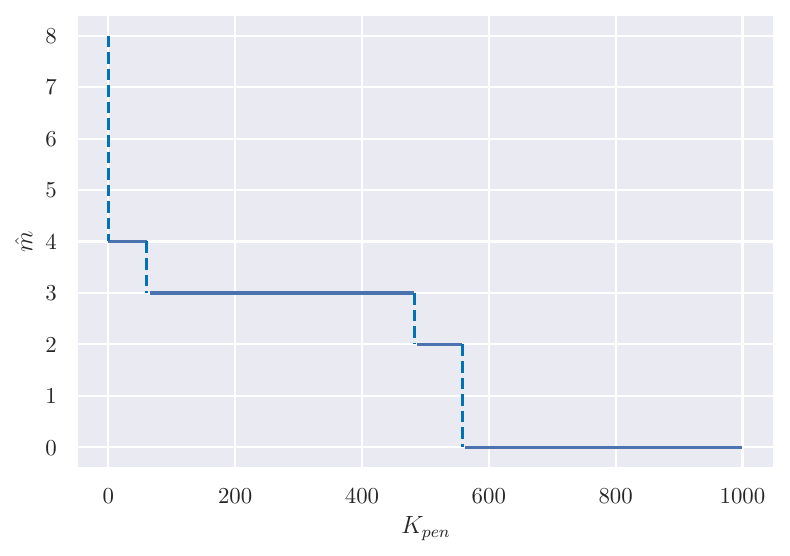} 
	    \caption{Selection of $K_{\textnormal pen}$ with the slope heuristics method.}
	    \label{fig:Kpen_selection_Y_sig_5}
    \end{minipage}
\end{figure}

For any $1 \leq i \leq n$, let $X_i: [0,1] \to \R^{d}$, $X_{i,t}= (X^1_{i,t},\dots,X^d_{i,t})$ be defined by
\begin{align} \label{eq:simulations_polysinus}
X^k_{i,t} =\alpha^k_{i,1} + 10\alpha^k_{i,2} \sin \Big(\frac{2 \pi t}{\alpha^k_{i,3}}\Big) + 10(t-\alpha^k_{i,4})^3, \quad 1 \leq k \leq d,
\end{align}
where the parameters $\alpha^k_{i, \ell}$, $1 \leq \ell \leq 4$ are sampled uniformly on $[0,1]$. Let $(t_0, \dots, t_{p-1})$ be a regular partition of $[0,1]$ of length $p$, the matrix of the path values \[\mathbf{x_i} = (x^k_{i,j})_{\substack{1 \leq k \leq d \\ 1 \leq j \leq p}} \in \R^{d \times p}\] is then a discretization of $X_i$ on $[0,1]$: $x^k_{i,j}=X^k_{i,t_j}$. It will cause no confusion to use the same notation $\mathbf{x_i}$ to denote the matrix of values of $X_i$ on the partition $(t_0, \dots, t_{p-1})$ and their piecewise linear interpolation. Fig.~\ref{fig:samples_X_polysinus_independent} shows one sample $\mathbf{x_i}$ with $p=100$ and $d=5$.

For any $m^\ast \in \N$, the output $Y_i$ is now defined as $Y_i=\langle \beta, S^{m^\ast}(\mathbf{x_i})\rangle + \varepsilon_i $, where $\varepsilon_i$ is a uniform random variable on $[-100,100]$ and $\beta$ is given by
\[ \beta_j=\frac{1}{1000}u_j, \quad 1\leq j \leq s_d(m^\ast), \]
where $u_j$ is sampled uniformly on $[0,1]$. Then, $m^\ast$ is estimated with the procedure described in Algorithm \ref{algo:procedure} for different sample sizes $n$. To select the constant $K_{\textnormal{pen}}$, we use the dimension jump method, that is we plot $\widehat{m}$ as a function of $K_{\textnormal{pen}}$, find the value of $K_{\textnormal{pen}}$ that corresponds to the first big jump of $\widehat{m}$ and fix $K_{\textnormal{pen}}$ to be equal to twice this value. For a recent account of the theory of slope heuristics, we refer the reader to the review by \cite{arlot2019minimal}. For example, for $m^\ast=5$,$d=2$, and $n=50$, plotting $\widehat{m}$ against $K_{\textnormal{pen}}$ yields Fig.~\ref{fig:Kpen_selection_Y_sig_5}. In this case, $K_{\textnormal{pen}}$ is selected at $100$.

\begin{figure}[ht]
 	\centering
	\includegraphics[width=.7\textwidth]{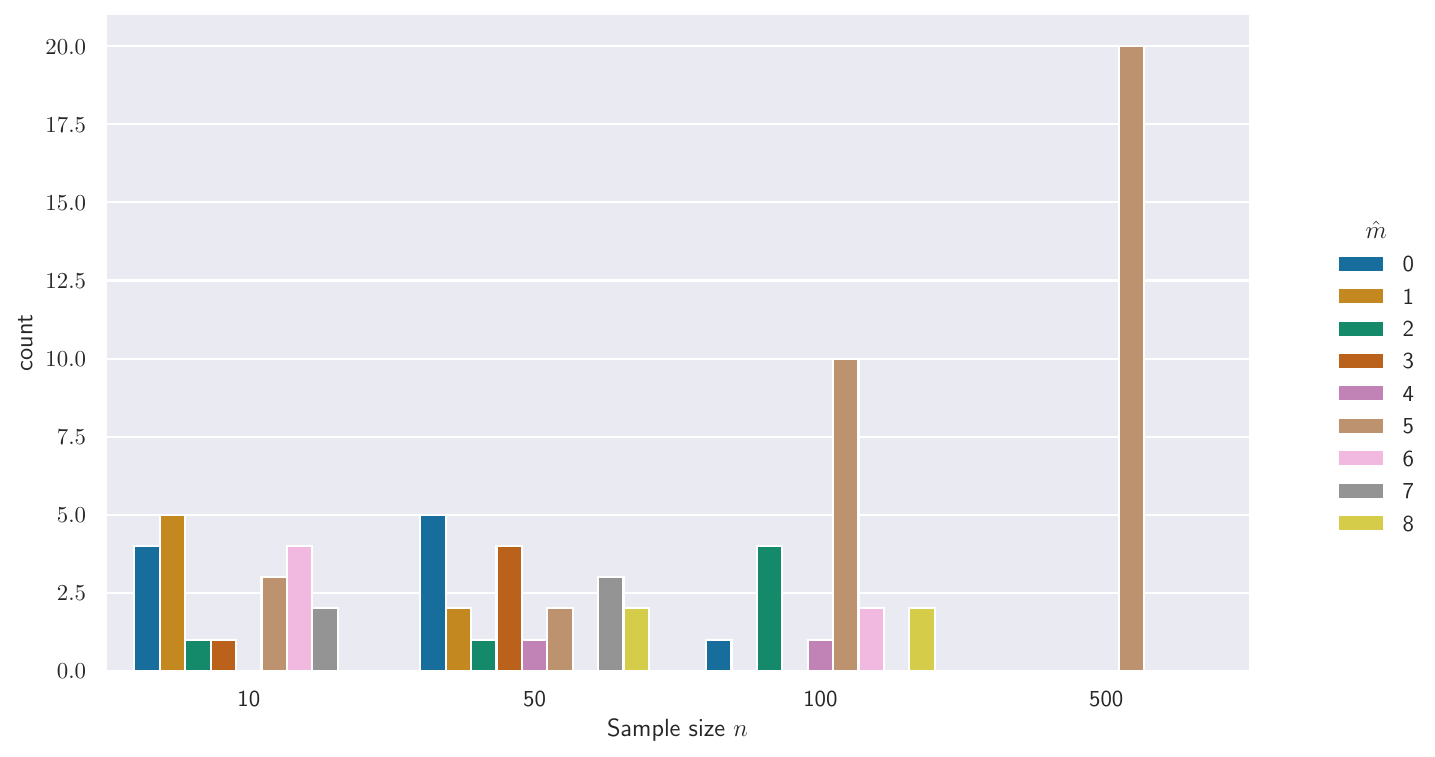} 
	\caption{Histogram of $\widehat{m}$ as a function of $n$ over 20 iterations. The functional predictors $X$ are simulated following \eqref{eq:simulations_polysinus} and the response $Y$ follows the linear model on signatures with $m^\ast = 5$. The hyperparameters are $\rho=0.4$ and $K_{\textnormal{pen}}=20$.}
	\label{fig:boxplots_mast_5}
\end{figure}

We fix $d=2$, $m^\ast=5$, and $K_{\textnormal{pen}} = 20$. For different sample sizes $n$, we iterate Algorithm \ref{algo:procedure} twenty times. In Fig.~\ref{fig:boxplots_mast_5}, a histogram of the values taken by $\widehat{m}$ is plotted against $n$. We can see that when $n$ increases, the estimator converges to the true value $m^\ast=5$. For $n=500$ we always pick $\widehat{m}=5$ over the twenty iterations.

\section{Experimental results} \label{sec:experimental_results}

Now that we have a complete procedure at hand, we demonstrate in this section its performance compared to canonical approaches in functional data analysis. We show in particular that it performs better in high dimensions, that is when $d$ is large.

We compare our model to the functional linear model with basis functions presented in Section \ref{sec:functional_linear_regression}, to functional principal component regression (fPCR), and to functional k-nearest neighbors regression. The first models are parametric linear models, while the k-nearest neighbors is nonlinear and nonparametric. Concerning the functional linear model, we consider two choices for the basis $\phi_1, \dots, \phi_K$, namely the B-Spline and Fourier basis \citep[see][]{ramsay2005functional}. Then, the approach consists in projecting the function $X: [0,1] \to \R^d$ onto the $\phi_i$s, coordinate by coordinate. The number $K$ of basis functions is selected via cross-validation (with a minimum of 4 and maximum of 14 for Fourier and B-Splines, and a minimum of 1 and a maximum of 6 for the fPCR). For the fPCR, we first smooth the functional covariates with 7 B-Splines. The number of neighbors is selected by cross-validation with a minimum of 1 and a maximum of 9. This procedure is implemented with the Python package \texttt{scikit-fda} \citep{ramos2019scikit}. In Subsections \ref{sec:simu_study_smooth_paths} and \ref{sec:simu_study_gaussian_processes}, since the focus is on the performance of the signature linear model and to simplify the computations, we select $\widehat{m}$ via cross-validation. For the real-world dataset of Subsection \ref{sec:air_quality}, it is estimated as described in the previous section.

\subsection{Smooth paths} \label{sec:simu_study_smooth_paths}

Our goal is to see the influence of the dimension $d$ on the quality of the different models: the signature linear model and the 3 linear functional models. To this end, we simulate some paths following model \eqref{eq:simulations_polysinus} and predict the average value of the path at the next time step. More precisely, let $(t_0, t_1, \dots, t_{p})$ be a partition of $[0,1]$ of length $p+1$, then we sample $X_i$ following \eqref{eq:simulations_polysinus} and let
\begin{align*}
\mathbf{x_i} = (X_{i, t_0}| \cdots| X_{i, t_{p-1}}) \in \R^{d \times p}, \quad Y_i = \frac{1}{d} \sum_{k=1}^d X^k_{i, t_{p}} + \varepsilon_i,
\end{align*}
where $\varepsilon_i$ are i.i.d uniform random variables on $[-1, 1]$. We let $d$ vary on a grid from $1$ to $11$, simulate some train and test data, and assess the performance of the model with the mean squared error (MSE) on the test set. We iterate the procedure 20 times, which gives, for each model (signature, Fourier, B-Spline, and fPCR), a boxplot of errors, shown in Fig.~\ref{fig:dimension_study_independent_sinus}. 

\begin{figure}[ht]
    \centering
    \includegraphics[width=.85\textwidth]{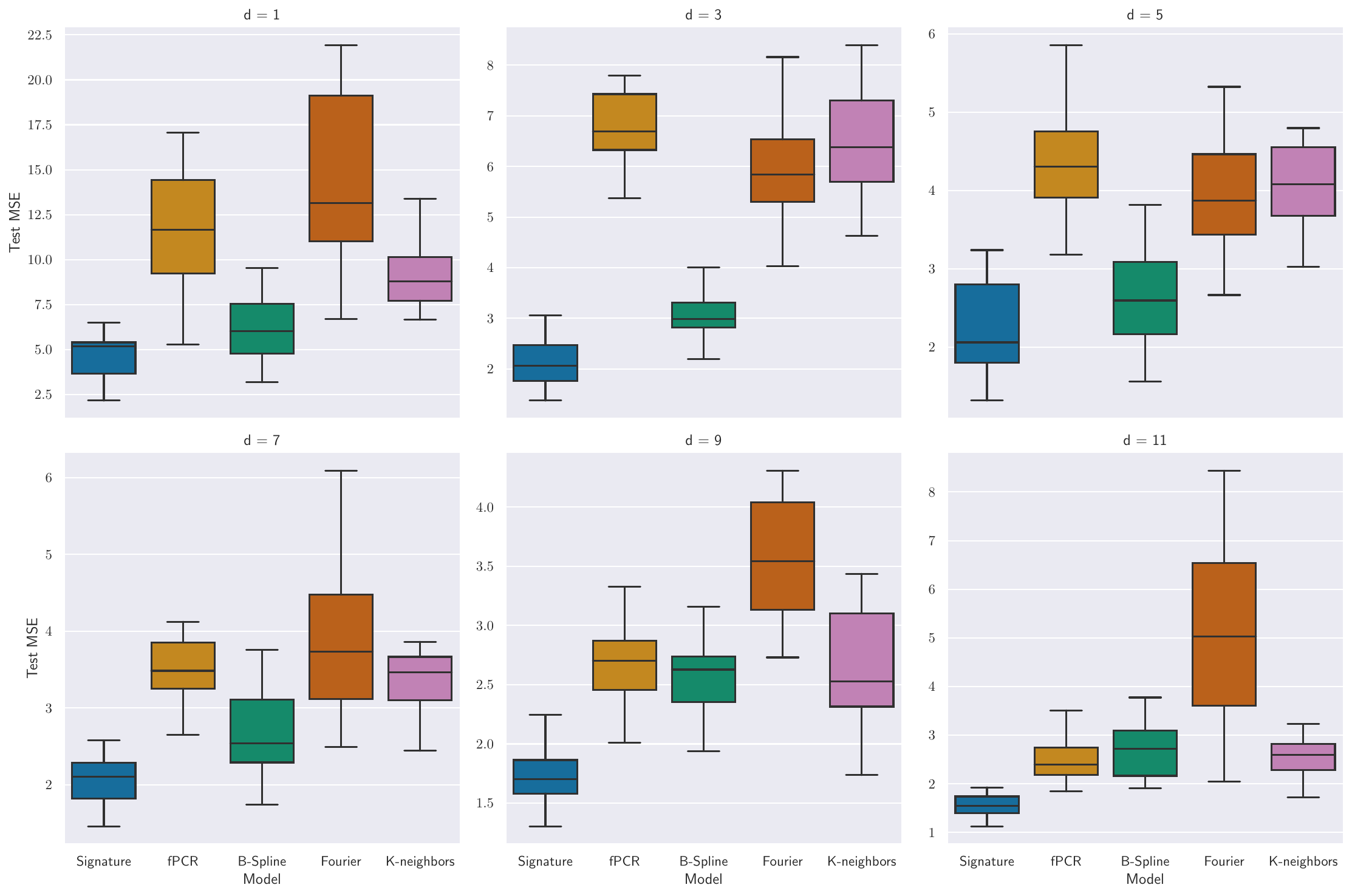} 
    \caption{Test MSE for different regression models when the inputs follow \eqref{eq:simulations_polysinus} and $Y$ is the mean response at the next time step.}
    \label{fig:dimension_study_independent_sinus}
\end{figure}

It is clear that when $d$ increases, the signature gets better relatively to the 4 other models. We can also note that the B-Spline basis performs best in low dimensions, which is not surprising since the data has a 3rd order polynomial term---see \eqref{eq:simulations_polysinus}. However, even though the B-Spline basis is particularly well-adapted to the data, it is outperformed by the signature linear model when the dimension becomes too large (starting from $d=7$).

\subsection{Gaussian processes} \label{sec:simu_study_gaussian_processes}

\begin{figure}[ht]
\centering
	\begin{minipage}{.4\textwidth}
	    \centering
	    \includegraphics[width=\textwidth]{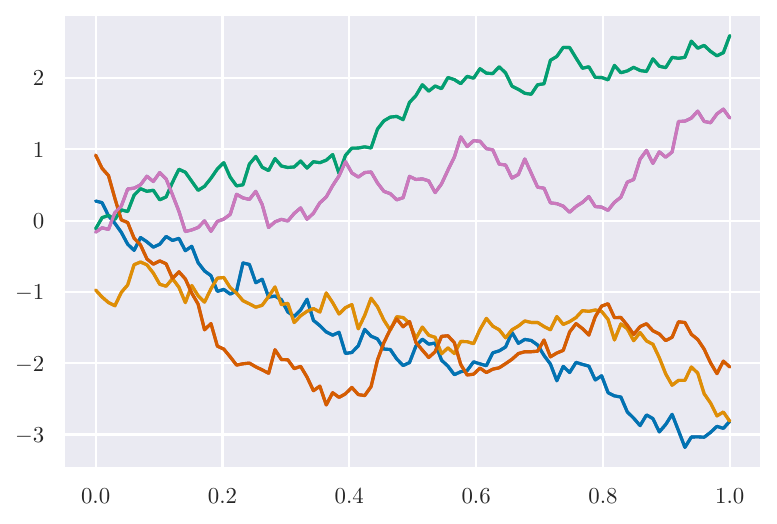} 
	    \caption{One sample $X$ from the Gaussian process model \eqref{eq:simulations_gaussian_processes} with $d=5$}
	    \label{fig:samples_gaussian_processes}
    \end{minipage}
    \begin{minipage}{.4\textwidth}
     \centering
	    \includegraphics[width=\textwidth]{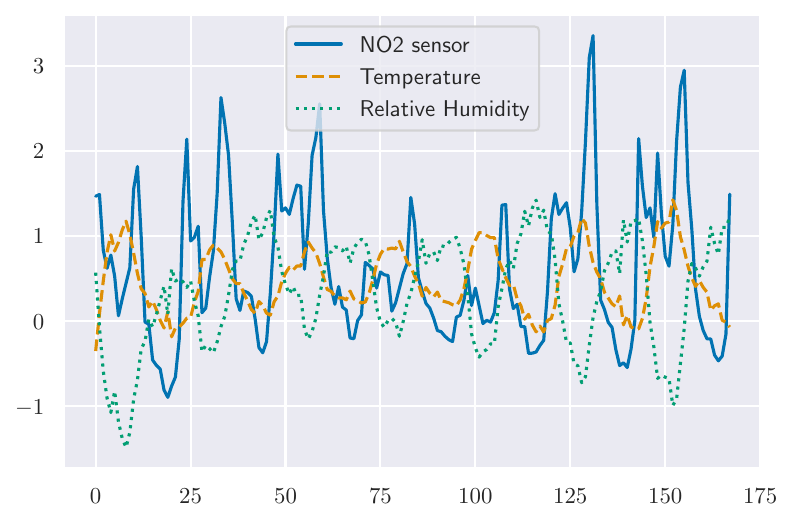} 
	    \caption{One sample from the Air Quality dataset.}
	    \label{fig:samples_air_quality}
    \end{minipage}
\end{figure}

We continue this simulation study with more complex paths: Gaussian processes. Let $d \geq 1$, $1 \leq i \leq n$, we define the path $X_i = (X^1_{i,t}, \dots, X^d_{i,t})_{t \in [0,1]}$ by 
\begin{equation} \label{eq:simulations_gaussian_processes}
X_{i,t}^k = \alpha_i^k t + \xi^k_{i,t}, \quad 1 \leq k \leq d, \quad t \in [0,1], 
\end{equation}
where $\alpha_i^k$ is sampled uniformly in $[-3,3]$ and $\xi^k_i$ is a Gaussian process with exponential covariance matrix (with length-scale 1). The response is the norm of the trend slope: $Y_i = \| \alpha_i\| + \varepsilon_i$, where $\varepsilon_i$ is uniformly sampled on $[-1,1]$. Fig.~\ref{fig:samples_gaussian_processes} shows a realization of $X_i$ with $d=5$.

\begin{figure}[ht]
    	\centering
	    \includegraphics[width=.85\textwidth]{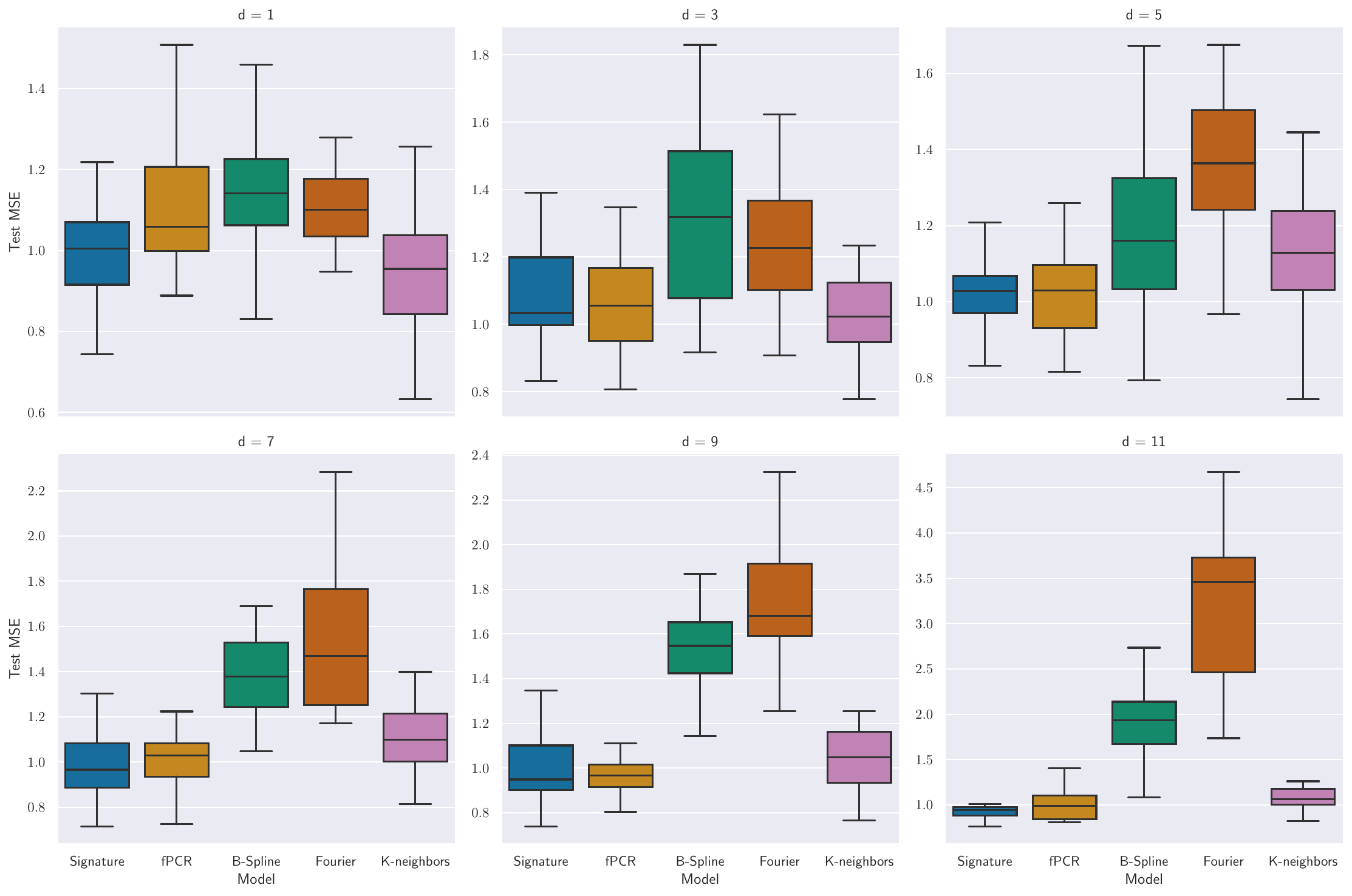} 
	    \caption{Test MSE for different regression models when the inputs are gaussian processes with a random linear trend, as defined by \eqref{eq:simulations_gaussian_processes}, and the response is the norm of the trend slope.}
	    \label{fig:dimension_study_gaussian_processes}
\end{figure}

We vary the dimension $d$ on the same grid as before and iterate the whole procedure 20 times, which gives the results in Fig.~\ref{fig:dimension_study_gaussian_processes}. We can see that for these more complicated paths, the signature is better than the 3 linear models even for $d=1$, but similar to the k-neighbors regression. The difference in performance with B-Spline and Fourier basis increases a lot with $d$, wheras the k-neighbors model is quite stable.

\subsection{Air quality dataset}\label{sec:air_quality}

We conclude this section with a study of the UCI ``Air Quality Data Set'' \citep{de2008field}. The data contains the hourly averaged response from 5 metal oxide chemical sensors recorded in a polluted area in Italy during a year (from March 2004 to February 2005). Ground truth concentrations are also included, together with temperature and humidity values. We restrict our analysis to the study of the concentration of nitrogen dioxide (N02), and more precisely to the prediction of the ground truth value of NO2 at the next hour. We consider two situations for the predictor function $X$: a univariate and a multivariate case. In the univariate case, we are given the values of the sensor recording the concentration of NO2 during the previous 7 days. In this case, the data is in dimension $d=1$ and sampled at $p=168$ values. In the multivariate case, we add the information of temperature and relative humidity to $X$, making it a path in dimension $d=3$. We show in Fig.~\ref{fig:samples_air_quality} one sample in the multivariate case (in the univariate case, $X$ consists only of the blue solid curve).

 \begin{figure}[ht]
        \centering
        \begin{subfigure}[b]{0.45\textwidth}
            \includegraphics[width=\textwidth]{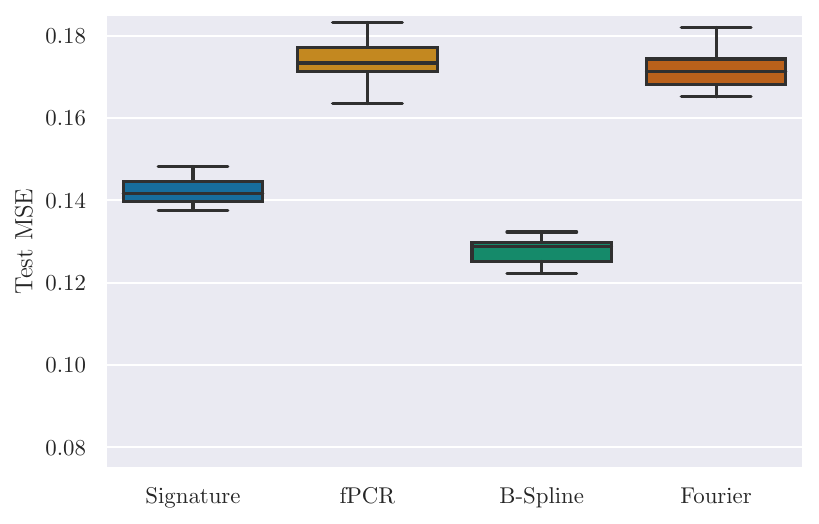}
            \caption{Univariate case.}
            \label{fig:results_air_quality_univariate}
        \end{subfigure}%
                \begin{subfigure}[b]{0.45\textwidth}
            \includegraphics[width=\textwidth]{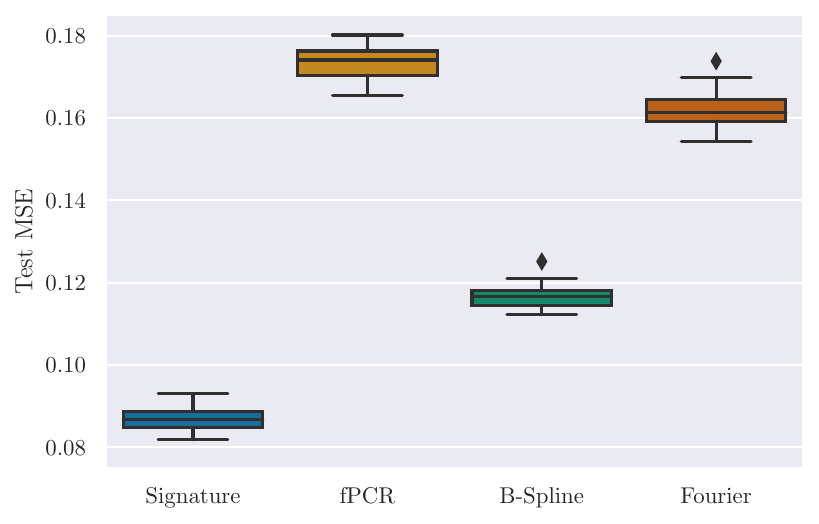}
            \caption{Multivariate case.}
            \label{fig:results_air_quality_multivariate}
        \end{subfigure}%
       \caption{Test MSE for different regression models for the Air Quality dataset.}
       \label{fig:results_air_quality}
\end{figure}

We perform 20 random train/test splits and show in Fig.~\ref{fig:results_air_quality} a boxplot of the test MSE for each model. We do not consider the k-neighbors regression due to its prohibitive running time for the sample size of this dataset (6156 training samples and 3033 test samples). Indeed, the other models take a few minutes to run while the k-neighbors regression takes two hours. We can see that in the univariate case, the B-Splines perform best. However, when more information is taken into account, that is, in the multivariate case, the signature model has the smallest error. The error of the three other models almost does not change when information about temperature and humidity is added, whereas the error of the signature linear model is divided by 2. We conclude that signatures can extract relevant information from multivariate time series. It should be noted that this type of data is increasingly common in modern applications, as the capabilities for recording and storing data are only getting better.

 \begin{figure}[ht]
        \centering
        \includegraphics[width=.7\textwidth]{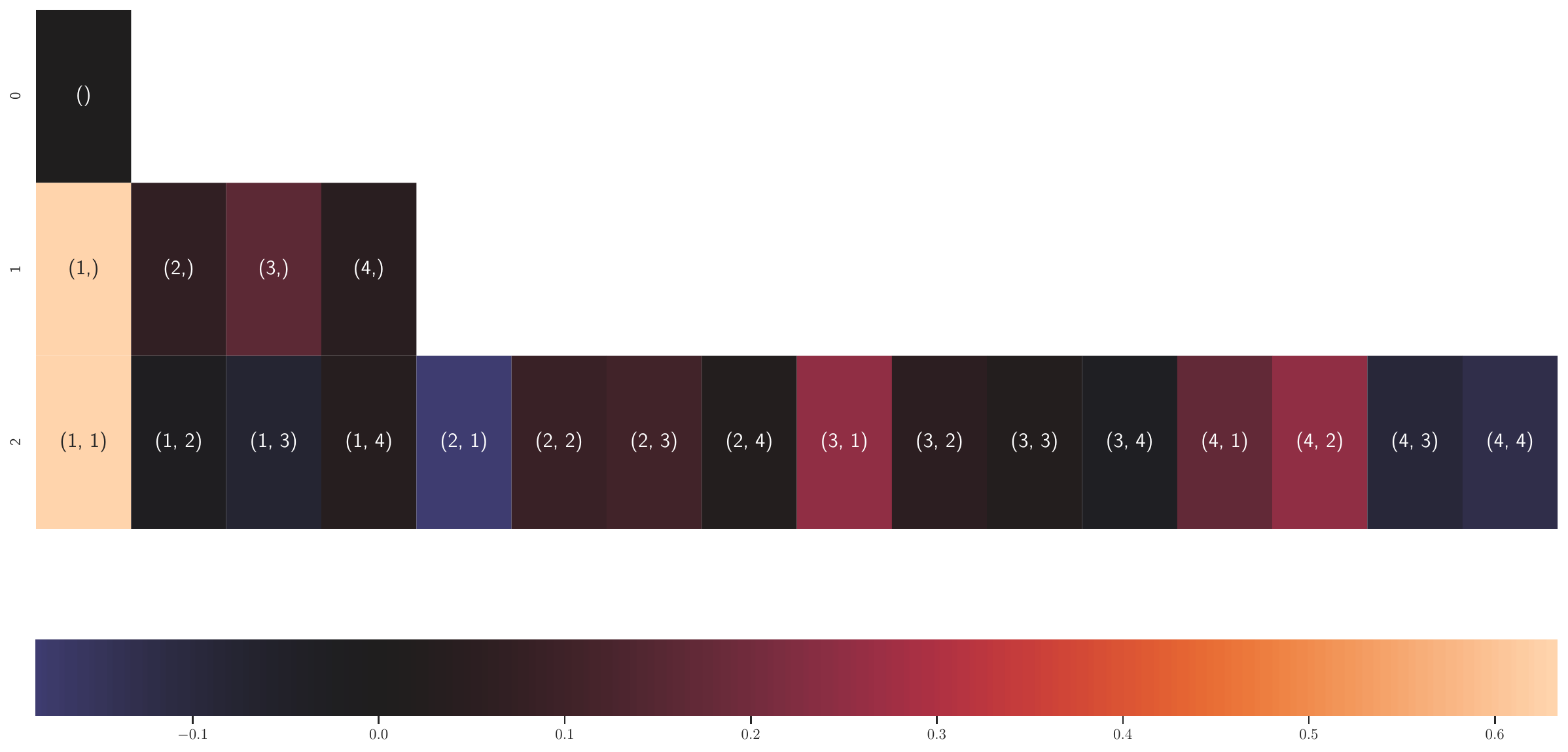}
        \caption{Heatmap of the first coefficients of $\widehat{\beta}_{\widehat{m}}$ for the Air quality dataset with a truncation order of 4. The vertical axis represents the order of the coefficients: on top the coefficient of order 0, then the 4 coefficients of order 1, then the 16 coefficients of order 2. The color corresponds to the value of the coefficient.}
        \label{fig:air_quality_coefficients}
\end{figure}

To conclude, we represent in Fig.~\ref{fig:air_quality_coefficients} the values of the regression vector $\widehat{\beta}_{\widehat{m}}$ to illustrate its interpretation. We observe that the two largest coefficients are the ones corresponding to $S^{(1)}(X)$ and $S^{(1, 1)}(X)$: they both correspond to the variation in NO2 concentration during the period (last value minus initial value). It is therefore not surprising that this is a key quantity to predict the concentration of NO2 at the next hour. We can also comment on the large absolute value of some coefficients of order 2, for example, the one corresponding to $S^{(2,1)}(X)$. The value $S^{(2,1)}(X)$ is the area under the curve (Temperature, NO2), as explained in Fig.~\ref{fig:geometric_signature}. The corresponding coefficient, therefore, contains information about the importance of the joint evolution of Temperature and NO2 to predict future concentration. For example, if it is positive, it means that a common increase in Temperature and NO2 will give rise to a larger concentration of NO2. In other words, there is an interaction between Temperature and NO2 concentration. A similar analysis can be done for the curve (Humidity, NO2), which corresponds to the coefficient (3,1). Finally, the large value of the coefficient corresponding to $S^{(4,1)}(X)$, which is equal to the area under the curve (Time, NO2) is also not surprising: it counts the total quantity of concentration of NO2 during the period. 

To conclude, the coefficients obtained with the signature linear model have a geometric interpretation, which is often valuable for practical applications. Contrary to the coefficients in traditional functional linear models, they are global measures of interaction between coordinates: there is no time-specific interpretation as there would be for $\beta(t)$ in \eqref{eq:functional_linear_model}. We refer to \citet{giusti2020iterated} for more details on the interpretation of signatures, in particular as a measure of causality between different coordinates.

\section{Conclusion and perspectives}

In this paper, we have provided a complete and ready-to-use methodology to implement the signature linear model. This led us to define a consistent estimator of the signature truncation order. We show on both simulated and real-world datasets that this model performs better than traditional functional linear models when the functional data is vector-valued, especially in high dimensions.
 
The signature is a flexible tool for summarizing multidimensional time series and can be used in various contexts. This study is just a first step towards understanding how it should be used in a statistical setting and there are a lot of potential extensions. First, we restricted our study to the setting of linear regression, however, signatures are just as relevant in classification or unsupervised learning settings. Moreover, the problem of the high dimension of the regression coefficient, due to its exponential dependence on $m$, is the major limitation of the signature linear model. In this article, we dealt with it by carefully choosing the truncation order. However, this is not the only option. For example, regularization approaches that induce a sparsity pattern on this coefficient, or the use of a related lower-dimensional object called the logsignature, are two interesting directions.

\section{Proofs} \label{sec:proofs}
\begin{flushleft}
\textbf{Proof of Theorem \ref{th:main_bound_hatm}}
\end{flushleft}

This section is devoted to the proof of Theorem \ref{th:main_bound_hatm}. We will use extensively results from \cite{van2014probability}. The next two lemmas first show that it is sufficient to obtain a uniform tail bound on the risk to control the convergence of $\widehat{m}$. 

\begin{lemma}
	\label{lemma:L_uniform_risk}
	For any $ m \in \N$,
	\[\big|\widehat{L}_n(m) - L(m) \big| \leq \underset{\beta \in B_{m,\alpha}}{\sup}  \big| \widehat{\risk}_{m,n}(\beta) - \risk_m(\beta) \big|.\]
\end{lemma}

\begin{proof}[\textbf{\upshape Proof:}]
	Introducing $\widehat{\risk}_{m,n}(\beta_m^\ast)$ yields
	\begin{align*}
	\widehat{L}_n(m) - L(m)&= \widehat{\risk}_{m,n}(\widehat{\beta}_{m}) -\risk_m(\beta^\ast_{m}) = \widehat{\risk}_{m,n}(\widehat{\beta}_{m}) -\widehat{\risk}_{m,n}(\beta^\ast_{m}) + \widehat{\risk}_{m,n}(\beta^\ast_{m})-\risk_m(\beta^\ast_{m}).
	\end{align*}
	Since  $\widehat{\beta}_{m}$ minimises $\widehat{\risk}_{m,n}$ over $B_{m,\alpha}$, $\widehat{\risk}_{m,n}(\widehat{\beta}_{m}) -\widehat{\risk}_{m,n}(\beta^\ast_{m}) \leq 0$, which gives
	\[\widehat{L}_n(m) - L(m)\leq \widehat{\risk}_{m,n}(\beta^\ast_{m})-\risk_m(\beta^\ast_{m}) \leq \underset{\beta \in B_{m,\alpha}}{\sup} | \widehat{\risk}_{m,n}(\beta) - \risk_m(\beta) |.\]
	In the same manner, $L(m)-\widehat{L}_n(m)  \leq \underset{\beta \in B_{m,\alpha}}{\sup} | \widehat{\risk}_{m,n}(\beta) - \risk_m(\beta) |,$ which proves the lemma.
\end{proof}

\begin{lemma}
	\label{lemma:bound_prob_k_geq_hatm}
	For any $m>m^\ast$, $\prob(\widehat{m}=m) \leq \prob \big(2 \sup_{\beta \in B_{m,\alpha}} |\widehat{\risk}_{m,n}(\beta) - \risk(\beta)| \geq \textnormal{pen}_n(m) -\textnormal{pen}_n(m^\ast) \big).$
\end{lemma}
\begin{proof}[\textbf{\upshape Proof:}]
	 For any $m \in \N$, 
	\begin{align*}
	\label{eq:prob_hatm_equals_k}
	\prob (\widehat{m}=m) & \leq \prob \left(\widehat{L}_n(m) + \text{pen}_n(m) \leq \widehat{L}_n(m^\ast) + \text{pen}_n(m^\ast) \right) \nonumber  =\prob  \left(\widehat{L}_n(m^\ast) -\widehat{L}_n(m) \geq \text{pen}_n(m) -\text{pen}_n(m^\ast)\right). 
	\end{align*}
	Recall that, by definition of model \eqref{eq:model_def}, $m \mapsto L(m)$ is a decreasing function and that its minimum is attained at $m=m^\ast$. Therefore, for any $m \in \N$, $L(m^\ast) \leq L(m)$, and Lemma \ref{lemma:L_uniform_risk} yields
	\begin{align*}
	\widehat{L}_n (m^\ast) - \widehat{L}_n(m)& = \widehat{L}_n(m^\ast) - L(m^\ast) + L(m^\ast) - L(m) + L(m) - \widehat{L}_n(m)  \leq \widehat{L}_n(m^\ast) - L(m^\ast) + L(m) - \widehat{L}_n(m) \nonumber \\
	&\leq \textnormal{sup}_{\beta \in B_{m^\ast,\alpha}} | \widehat{\risk}_{m,n}(\beta) - \risk_m(\beta) |  + \textnormal{sup}_{\beta \in B_{m,\alpha}} | \widehat{\risk}_{m,n}(\beta) - \risk_m(\beta) |.
	\end{align*}
	For $m>m^\ast$, $B_{m^\ast, \alpha} \subset B_{m,\alpha}$, which gives $\widehat{L}_n (m^\ast) - \widehat{L}_n(m) \leq 2 \textnormal{sup}_{\beta \in B_{m,\alpha}} | \widehat{\risk}_{m,n}(\beta) - \risk_m(\beta) |,$ and the proof is complete.
\end{proof}
From now on, we denote by $Z_{m,n}$ the centered empirical risk for signatures truncated at $m$: for any $\beta \in B_{m,\alpha}$,
\[Z_{m,n}(\beta) =\widehat{\risk}_{m,n}(\beta) - \risk_m(\beta)= \frac{1}{n} \sum_{i=1}^{n} \big(Y_i - \big\langle \beta, S^m(X_i) \big\rangle \big)^2 - \esp \big(Y- \big\langle \beta, S^m(X) \big\rangle \big)^2. \]
We will now derive a uniform tail bound on $Z_{m,n}(\beta)$, which is the main result needed to prove Theorem \ref{th:main_bound_hatm}. In a nutshell, we show that $(Z_{m,n}(\beta))_{\beta \in B_{m,\alpha}}$ is a subgaussian process for some appropriate distance, and then use a chaining tail inequality \citep[][Theorem 5.29]{van2014probability} on $Z_{m,n}$.
\begin{lemma}
	\label{lemma:R_n_subgaussian}
	Under the assumptions $(H_{\alpha})$ and $(H_K)$, for any $m \in \N$, the process $ \left( Z_{m,n}(\beta) \right)_{\beta \in B_{m,\alpha}}$ is subgaussian for the distance 
	\begin{equation}
	\label{eq:def_D}
	D(\beta, \gamma) = \frac{K}{\sqrt{n}}\| \beta - \gamma\|,
	\end{equation}
where the constant $K$ is defined by \eqref{eq:def_K}.
\end{lemma}
\begin{proof}[\textbf{\upshape Proof:}]
	By definition, it is clear that $\esp Z_{m,n}(\beta)=0$ for any $\beta \in B_{m,\alpha}$. Let $\ell_{(X,Y)} \colon B_{m,\alpha} \to \R$ be given by $\ell_{(X,Y)}(\beta) = \big(Y - \big\langle \beta, S^m(X) \big\rangle \big)^2.$
	We first prove that $\ell_{(X,Y)}$ is $K$-Lipschitz. For any $\beta, \gamma \in B_{m,\alpha}$,
	 \begin{align*}
	 |\ell_{(X,Y)}(\beta) - \ell_{(X,Y)}(\gamma)| & =  \big|  \big(Y - \big\langle \beta, S^m(X) \big\rangle \big)^2 - \big(Y - \big\langle \gamma, S^m(X) \big\rangle \big)^2 \big|  \\
	 & \leq 2 \max \Big(\big|Y - \big\langle \beta, S^m(X) \big\rangle \big|, \big|Y - \big\langle \gamma, S^m(X) \big\rangle \big|\Big) \times \big|\big\langle \beta - \gamma, S^m(X) \big\rangle \big|  \\
	 & \quad  (\text{because } |a^2 - b^2| \leq 2 \max(|a|,|b|)|a-b|)  \\
	 &\leq 2 \max \Big(\big|Y - \big\langle \beta, S^m(X) \big\rangle \big|, \big|Y - \big\langle \gamma, S^m(X) \big\rangle \big|\Big) \times \big\| S^m(X)\big\| \left\|\beta - \gamma\right\|    \\
	 &  \quad (\text{by the Cauchy-Schwartz inequality} ).
	 \end{align*}
	Moreover, by the triangle inequality and Cauchy-Schwartz inequality,
	 \begin{align*}
	 \big|Y - \big\langle \beta, S^m(X) \big\rangle \big| &\leq |Y| + \big\|S^m(X) \big\| \| \beta \|   \leq K_Y + \alpha\big\|S^m(X) \big\| ,
	 \end{align*}
	and, by Proposition  \ref{prop:exp_decay}, $\big\| S^m(X)\big\| \leq e^{\| X \|_{TV}} \leq e^{K_X}.$ Consequently, $\big|Y - \big\langle \beta, S^m(X) \big\rangle \big| \leq K_Y + \alpha e^{K_X},$ and
	 \[	\big| \ell_{(X,Y)}(\beta) - \ell_{(X,Y)}(\gamma) \big| \leq 2 \big(K_Y + \alpha e^{K_X} \big) e^{K_X}\| \beta - \gamma\| = K\| \beta - \gamma\| .\]
 Therefore, by Hoeffding's lemma \citep[][Lemma 3.6]{van2014probability}, $\ell_{(X,Y)}(\beta) - \ell_{(X,Y)}(\gamma) $ is a subgaussian random variable with variance proxy $K^2 \|\beta - \gamma \|^2$, which gives, for $\lambda \geq 0$,
	\begin{align*}
	&\esp \exp \bigg( \lambda \Big(\ell_{(X,Y)}(\beta) - \ell_{(X,Y)}(\gamma) - \esp\big(\ell_{(X,Y)}(\beta) - \ell_{(X,Y)}(\gamma) \big) \Big) \bigg) \leq \exp \bigg(\frac{\lambda^2 K^2 \left\| \beta - \gamma \right\|^2}{2}\bigg).
	\end{align*}
	From this, it follows that
	\begin{align*}
	\esp e^{\lambda \big(Z_{m,n}(\beta) - Z_{m,n}(\gamma) \big)}&= \esp \exp \bigg(\frac{\lambda}{n} \sum_{i=1}^{n} \ell_{(X_i,Y_i)}(\beta) - \ell_{(X_i,Y_i)}(\gamma) - \esp\Big(\ell_{(X_i,Y_i)}(\beta) - \ell_{(X_i,Y_i)}(\gamma) \Big) \bigg)\\
	& = \prod_{i=1}^{n} \esp \exp \bigg(\frac{\lambda}{n}\Big( \ell_{(X_i,Y_i)}(\beta) - \ell_{(X_i,Y_i)}(\gamma) - \esp\big(\ell_{(X_i,Y_i)}(\beta)- \ell_{(X_i,Y_i)}(\gamma) \big) \Big) \bigg) \\
	& \leq \exp \Big(\frac{\lambda^2 K^2 \left\| \beta - \gamma \right\|^2}{2n} \Big)  = \exp \Big(\frac{\lambda^2 D(\beta, \gamma)^2}{2}\Big),\\
	\end{align*}
	where $D(\beta,\gamma) =\frac{K\| \beta-\gamma\|}{\sqrt{n}}$, which completes the proof.
\end{proof}
We can now derive a maximal tail inequality for $Z_{m,n}(\beta)$.
\begin{proposition}
	\label{prop:unif_risk_bound}
	Under the assumptions $(H_{\alpha})$ and $(H_K)$, for any $m \in \N$, $x>0$, $\beta_0 \in B_{m,\alpha}$, 
	\[\prob \bigg( \underset{\beta \in B_{m,\alpha}}{\sup} Z_{m,n}(\beta) \geq 108 \sqrt{\pi} K \alpha \sqrt{\frac{s_d(m)}{n}} +Z_{m,n}(\beta_0) +x \bigg) \leq  36 \exp \Big(- \frac{x^2 n }{ 144 K^2 \alpha^2}\Big), \]
	where the constant $K$ is defined by \eqref{eq:def_K}.
\end{proposition}

\begin{proof}[\textbf{\upshape Proof:}]
	By Lemma \ref{lemma:R_n_subgaussian}, $Z_{m,n}$ is a subgaussian process for $D$, defined by \eqref{eq:def_D}. So, we may apply Theorem 5.29 of \cite{van2014probability} to $Z_{m,n}$ on the metric space $(B_{m,\alpha},D)$:
	\begin{align*}
	&\prob \Big(\underset{\beta \in B_{m,\alpha}}{\sup}Z_{m,n}(\beta) - Z_{m,n}({\beta_0}) \geq 36 \int_{0}^{\infty} \sqrt{\log(N(\varepsilon,B_{m,\alpha},D))} d\varepsilon  + x \Big) \leq  36 \exp \Big(- \frac{x^2 n }{ 36 \times 4 K^2 \alpha^2} \Big) ,
	\end{align*}
	where $N(\varepsilon,B_{m,\alpha},D)$ is the $\varepsilon$-covering number of $B_{m,\alpha}$ with respect to $D$, and where we use that 
	\[\text{diam}(B_{m,\alpha}) =  \frac{2K \alpha}{\sqrt{n}}.\] 
	Moreover, $N(\varepsilon,B_{m,\alpha},D) = N(\frac{\sqrt{n}}{K} \varepsilon, B_{m, \alpha}, \| \cdot\| )$ , and so, by Lemma 5.13 of \citet[][]{van2014probability},

	\[ N(\varepsilon,B_{m,\alpha},D) \leq \left( \frac{3 K \alpha}{ \sqrt{n}\varepsilon}\right)^{s_d(m)} \text{ if } \varepsilon < \frac{K \alpha}{\sqrt{n}},\]
	 and $N(\varepsilon,B_{m,\alpha},D) =1$ otherwise. Therefore, 
	\begin{align}
	\label{eq:entropy_integral}
	\int_{0}^{\infty} \sqrt{\log(N(\varepsilon,B_{m, \alpha},D))} d\varepsilon &=\int_{0}^{\frac{K \alpha}{\sqrt{n}}} \sqrt{\log(N(\varepsilon,B_{m,\alpha},D))} d\varepsilon \nonumber \\
	& \leq\int_{0}^{\frac{K \alpha}{\sqrt{n}}} \sqrt{s_d(m)\log \left(\frac{3K \alpha}{ \sqrt{n}\varepsilon} \right)} d\varepsilon \nonumber \\
	& \leq 3 K\alpha \sqrt{\frac{s_d(m)}{n}} \int_0^\infty 2 x^2 \exp \left({-x^2} \right) dx  =  3 K\alpha \sqrt{\frac{s_d(m)}{n}} \sqrt{\pi},
	\end{align}
	where in the second inequality we use the change of variable $x = \sqrt{\log \left(\frac{2 K \alpha}{ \sqrt{n}\varepsilon}\right)}  $. 
\end{proof}

Since $\prob (\widehat{m} \neq m^\ast) = \prob (\widehat{m} > m^\ast) + \prob (\widehat{m}< m^\ast), $ we divide the proof into two cases. Let us first consider $m > m^\ast$ in the next proposition.
\begin{proposition}
	\label{prop:cvg_hatm_geq_mast}
		Let $0 < \rho < \frac{1}{2}$, and $ \textnormal{pen}_n(m)$ be defined by \eqref{eq:def_pen}: $\textnormal{pen}_n(m) = K_{\textnormal{pen}} n^{-\rho} \sqrt{s_d(m)}.$ Let $n_1$ be the smallest integer satisfying
		\begin{equation}
		\label{eq:condition_n_1}
		n_1 \geq \bigg( \frac{432 \sqrt{\pi} K \alpha\sqrt{s_d(m^\ast+1)}}{K_{\textnormal{pen}}(\sqrt{s_d(m^\ast+1)} - \sqrt{s_d(m^\ast) }) } \bigg)^{1/(\frac{1}{2} - \rho)}.
		\end{equation}
		Then, under the assumptions $(H_\alpha)$ and $(H_K)$, for any $m> m^\ast$, $n \geq n_1$,
		\begin{equation*}
		\label{eq:bound_prob_hatm_eq_k_geq_mast}
		\prob \left(\widehat{m} = m \right)\leq 74 \exp \big(-C_3 (n^{1-2\rho} +s_d(m) ) \big),
		\end{equation*}
		 where the constant $C_3$ is defined by 
		 \begin{equation*}
		 \label{eq:def_C_3}
		 C_3 =\frac{K_{\textnormal{pen}}^2 d^{m^\ast+1}}{128s_d(m^\ast+1)(72K^2\alpha^2 + K_Y^2)}.
		 \end{equation*}
\end{proposition}
\begin{proof}[\textbf{\upshape Proof:}]
Let
 \[u_{m,n} = \frac{1}{2} \big( \text{pen}_n(m) - \text{pen}_n(m^{\ast}) \big) = \frac{K_{\textnormal{pen}}}{2} n^{-\rho} \Big(\sqrt{s_d(m)}-\sqrt{s_d(m^\ast)}\Big). \]
As $m \mapsto \text{pen}_n(m)$ is increasing in $m$, it is clear that $u_{m,n} > 0$ for any $m > m^\ast$. From Lemma \ref{lemma:bound_prob_k_geq_hatm}, we see that
\begin{align*}
\prob \left( \widehat{m} = m \right) &\leq  \prob \Big( \underset{\beta \in B_{m,\alpha}}{\sup} \left| Z_{m,n}(\beta)\right| > u_{m,n}\Big)=  \prob \Big( \underset{\beta \in B_{m,\alpha}}{\sup}  Z_{m,n}(\beta) > u_{m,n}\Big) +  \prob \Big( \underset{\beta \in B_{m,\alpha}}{\sup} \left(-Z_{m,n}(\beta)\right) > u_{m,n}\Big). 
\end{align*}
We focus on the first term of the inequality, the second can be handled in the same way since Proposition \ref{prop:unif_risk_bound} also holds when $Z_{m,n}(\beta)$ is replaced by $-Z_{m,n}(\beta)$. Let $\beta_0$ be a fixed point in $B_{m,\alpha}$ that will be chosen later, we have
\begin{align}
\label{eq:proof_sub_tail_ineq_1}
\prob \Big( \underset{\beta \in B_{m,\alpha}}{\sup} Z_{m,n}(\beta) > u_{m,n}\Big) & = \prob \Big( \underset{\beta \in B_{m,\alpha}}{\sup} Z_{m,n}(\beta) > u_{m,n}, \, Z_{m,n}(\beta_0) \leq \frac{u_{m,n}}{2} \Big) + \prob \Big( \underset{\beta \in B_{m,\alpha}}{\sup} Z_{m,n}(\beta) > u_{m,n}, \, Z_{m,n}(\beta_0) > \frac{u_{m,n}}{2} \Big) \nonumber\\
& \leq \prob \Big( \underset{\beta \in B_{m,\alpha}}{\sup} Z_{m,n}(\beta) > \frac{u_{m,n}}{2}+ Z_{m,n}(\beta_0)  \Big)  + \prob \Big(Z_{m,n}(\beta_0) > \frac{u_{m,n}}{2} \Big).
\end{align}
We treat each term separately. The first one is handled by Proposition \ref{prop:unif_risk_bound}. To this end, we need to ensure that $\frac{u_{m,n}}{2} - 108 K \alpha \sqrt{\frac{ \pi s_d(m)}{n}}$ is positive. By definition,
\begin{align*}
\frac{u_{m,n}}{2} - 108 K \alpha \sqrt{\frac{ \pi s_d(m)}{n}} &= \frac{K_{\textnormal{pen}}}{2}  n^{-\rho} \Big(\sqrt{s_d(m)} - \sqrt{s_d(m^\ast)}\Big)-108 K \alpha \sqrt{\frac{ \pi s_d(m)}{n}}  \\
& = \sqrt{s_d(m)} n^{-\rho} \frac{K_{\textnormal{pen}}}{2} \bigg( 1 - \sqrt{\frac{s_d(m^\ast)}{s_d(m)}} - \frac{2 \times 108 \sqrt{\pi} K \alpha}{K_{\textnormal{pen}}} n^{\rho - \frac{1}{2}} \bigg) .\\
& \geq \sqrt{s_d(m)} n^{-\rho} \frac{K_{\textnormal{pen}}}{2} \bigg( 1 - \sqrt{\frac{s_d(m^\ast)}{s_d(m^\ast+1)}} - \frac{216 \sqrt{\pi} K \alpha}{K_{\textnormal{pen}}} n^{\rho - \frac{1}{2}} \bigg) .\\
\end{align*}
Let $n_1 \in \N$ be such that
\begin{align*}
1 - \sqrt{\frac{s_d(m^\ast)}{s_d(m^\ast+1)}} - \frac{216 \sqrt{\pi} K \alpha}{K_{\textnormal{pen}}} n_1^{\rho - \frac{1}{2}} > \frac{1}{2} \bigg(1 - \sqrt{\frac{s_d(m^\ast)}{s_d(m^\ast+1)}} \bigg) \, \Leftrightarrow \,  n_1 > \bigg( \frac{432 \sqrt{\pi} K \alpha\sqrt{s_d(m^\ast+1)}}{K_{\textnormal{pen}}(\sqrt{s_d(m^\ast+1)} - \sqrt{s_d(m^\ast) }) } \bigg)^{1/(\frac{1}{2} - \rho)},
\end{align*}
then, for any $n \geq n_1$,
 \[\frac{u_{m,n}}{2} - 108 K \alpha \sqrt{\frac{\pi s_d(m)}{n}} \geq  \sqrt{s_d(m)} n^{-\rho}   \frac{K_{\textnormal{pen}}}{4} \bigg(1 - \sqrt{\frac{s_d(m^\ast)}{s_d(m^\ast+1)}} \bigg)  > 0 .\]
Hence, Proposition \ref{prop:unif_risk_bound} applied to $x = \frac{u_{m,n}}{2} - 108 \sqrt{\pi} K \alpha \sqrt{\frac{s_d(m)}{n}}$ now shows that, for $n \geq n_1$,
\begin{align}
\label{eq:bound_term_1}
\prob \Big( \underset{\beta \in B_{m,\alpha}}{\sup} Z_{m,n}(\beta) > \frac{u_{m,n}}{2} +Z_{m,n}(\beta_0)\Big) & \leq 36 \exp \bigg( - \frac{n}{144 K^2 \alpha^2} \bigg( \frac{u_{m,n}}{2} - 108 K\alpha \sqrt{\frac{\pi s_d(m)}{n}}\bigg)^2\bigg) \nonumber \\
& \leq  36 \exp \bigg( - \frac{ s_d(m) n^{1-2 \rho} K_{\textnormal{pen}}^2}{ 144 K^2 \alpha^2 \times 16} \bigg(1 - \sqrt{\frac{s_d(m^\ast)}{s_d(m^\ast+1)}} \bigg)^2 \bigg) \nonumber \\
& =36 \exp \Big( - \kappa_1 s_d(m) n^{1-2\rho} \Big),
\end{align}
where 
\[\kappa_1 =  \frac{K_{\textnormal{pen}}^2}{ 2304 K^2 \alpha^2} \bigg(1 - \sqrt{{\frac{s_d(m^\ast)}{s_d(m^\ast+1)}}} \bigg)^2 .\] 
We now turn to the second term of \eqref{eq:proof_sub_tail_ineq_1}. Since $\big|Y-\langle \beta_0, S^m(X) \rangle \big|^2\leq \big(K_Y  + \| \beta_0 \| e^{K_X} \big)^2$ a.s., Hoeffding's inequality yields, for $n \geq n_1$,
\begin{align}
\label{eq:bound_term_2}
\prob \Big(Z_{m,n}(\beta_0) > \frac{u_{m,n}}{2} \Big)& \leq \exp \bigg( - \frac{n u_{m,n}^2 }{8 \big(K_Y  + \| \beta_0 \| e^{K_X} \big)^2}\bigg) \nonumber \\
& = \exp \bigg( - \frac{ n^{1-2\rho} K_{\textnormal{pen}}^2 \left(\sqrt{s_d(m)}-\sqrt{s_d(m^\ast)} \right)^2}{32\big(K_Y  + \| \beta_0 \| e^{K_X} \big)^2}\bigg) \nonumber \\
& \leq \exp \bigg( - \frac{ n^{1-2\rho} K_{\textnormal{pen}}^2 s_d(m)}{32\big(K_Y  + \| \beta_0 \| e^{K_X} \big)^2} \bigg(1 - \sqrt{{\frac{s_d(m^\ast)}{s_d(m^\ast+1)}}} \bigg)^2 \bigg) \nonumber \\
& =\exp \left( -\kappa_2 n^{1-2\rho} s_d(m)\right) ,
\end{align}
where 
\[\kappa_2 = \frac{ K_{\textnormal{pen}}^2}{32\big(K_Y  + \| \beta_0 \| e^{K_X} \big)^2} \bigg(1 - \sqrt{{\frac{s_d(m^\ast)}{s_d(m^\ast+1)}}} \bigg)^2.\]
Combining \eqref{eq:bound_term_1} with \eqref{eq:bound_term_2}, we obtain
\begin{equation*}
\begin{split}
 \prob \Big( \underset{\beta \in B_{m,\alpha}}{\sup}  Z_{m,n}(\beta) > u_{m,n}\Big) \leq 36 \exp \big( - \kappa_1 n^{1-2\rho}s_d(m)  \big)+ \exp \big( - \kappa_2 n^{1-2\rho} s_d(m)\big)&\leq 37  \exp \big( -\kappa_3 n^{1-2\rho} s_d(m)\big) \\
&  \leq 37 \exp \Big( -  \frac{\kappa_3 }{2} \big( n^{1-2\rho} +s_d(m) \big) \Big) ,\\
\end{split}
\end{equation*}
where $\kappa_3 = \min (\kappa_1,\kappa_2)$. The same proof works for the process $\left(-Z_{m,n}(\beta)\right)$, and consequently

\begin{equation*}
\prob \left(\widehat{m} = m\right) \leq 2 \times 37 \exp \Big( -  \frac{\kappa_3 }{2} \big( n^{1-2\rho} +s_d(m) \big) \Big).
\end{equation*}
We are left with the task of choosing an optimal $\beta_0$. Since
\begin{align*}
\kappa_3 &= \min(\kappa_1,\kappa_2) = \frac{K_{\textnormal{pen}}^2}{32} \bigg(1 - \sqrt{{\frac{s_d(m^\ast)}{s_d(m^\ast+1)}}} \bigg)^2 \min \bigg(\frac{1}{72 K^2 \alpha^2},  \frac{1}{\big(K_Y  + \| \beta_0 \| e^{K_X} \big)^2} \bigg),
\end{align*}
it is clear that $\kappa_3$ is maximal at $\beta_0=0$, which yields
\begin{align*}
\kappa_3= \frac{K_{\textnormal{pen}}^2}{32} \bigg(1 - \sqrt{{\frac{s_d(m^\ast)}{s_d(m^\ast+1)}}} \bigg)^2\min \bigg(\frac{1}{72 K^2 \alpha^2},  \frac{1}{K_Y^2} \bigg).
\end{align*}
Noting that
\[\sqrt{s_d(m^\ast+1)} - \sqrt{s_d(m^\ast)} = \sqrt{d^{m^\ast+1} +s_d(m^\ast)} - \sqrt{s_d(m^\ast)} \geq \sqrt{\frac{d^{m^\ast+1}}{2}},\]
where we have used the fact that for $a,b \geq 0$, $\sqrt{a} + \sqrt{b} \geq \sqrt{2} \sqrt{a+b}$, letting
\[C_3 =\frac{1}{2}\times \frac{K_{\textnormal{pen}}^2 d^{m^\ast+1}}{64s_d(m^\ast+1)(72K^2\alpha^2 + K_Y^2)}\]
 completes the proof.
\end{proof}

To treat the case $m<m^\ast$, we need a rate of convergence of $\widehat{L}_n$. This can be obtained with arguments similar to the previous proof.

\begin{proposition}
	\label{prop:cvg_hatL}
	For any $\varepsilon > 0$, $m \in \N$, let $n_2 \in \N$ be the smallest integer such that
	\begin{equation}
	\label{eq:condition_n_2}
	n_2 \geq \frac{432^2 K^2 \pi  \alpha^2 s_d(m)}{\varepsilon^2}.
	\end{equation}
	Then, for any $n \geq n_2$,
	\begin{equation*}
	\prob \big( |\widehat{L}_n(m) - L(m)| > \varepsilon \big) \leq 74 \exp \big(-C_4n \varepsilon^2 \big),
	\end{equation*}
	where the constant $C_4$ is defined by
	\begin{equation}
	\label{eq:def_C_4}
	C_4 = \Big(2( 1152 K^2 \alpha^2 + K_Y^2) \Big)^{-1}.
	\end{equation}
\end{proposition}
\begin{proof}[\textbf{\upshape Proof:}]
	By Lemma \ref{lemma:L_uniform_risk}, 
	\begin{align*}
	\prob \big( | \widehat{L}_n(m) - L(m) |> \varepsilon \big) &\leq \prob \Big( \underset{\beta \in B_{m,\alpha}}{\sup}|Z_{m,n}(\beta) | > \varepsilon \Big) = \prob \Big( \underset{\beta \in B_{m,\alpha}}{\sup} Z_{m,n}(\beta) > \varepsilon \Big) + \prob \Big( \underset{\beta \in B_{m,\alpha}}{\sup} \left(-Z_{m,n}(\beta) \right) > \varepsilon \Big).
	\end{align*}
	 Let us fix $\beta_0 \in B_{m,\alpha}$, we can now proceed as in Proposition \ref{prop:cvg_hatm_geq_mast}. Since, for $n \geq n_2$,
	 \[\frac{\varepsilon}{2} -108K \alpha \sqrt{\frac{ \pi s_d(m)}{n}} > \frac{\varepsilon}{4} > 0  ,\] 
	 Hoeffing's inequality and Proposition \ref{prop:unif_risk_bound} show that
	\begin{align*}
	\prob \Big( \underset{\beta \in B_{m,\alpha}}{\sup} Z_{m,n}(\beta) > \varepsilon \Big) & \leq \prob \Big( \underset{\beta \in B_{m,\alpha}}{\sup} Z_{m,n}(\beta) > \frac{\varepsilon}{2} + Z_{m,n}(\beta_0) \Big)  + \prob \Big(Z_{m,n}(\beta_0) > \frac{\varepsilon}{2} \Big) \\
	& \leq 36 \exp \bigg( - \frac{n}{144 K^2\alpha^2} \Big( \frac{\varepsilon}{2}-108 K \alpha \sqrt{\frac{ \pi s_d(m)}{n}} \Big)^2\bigg) + \exp \bigg(- \frac{n\varepsilon^2}{ 2 \big(K_Y + \| \beta_0 \| e^{K_X}\big)^2}\bigg) \\
	& \leq 36 \exp \Big( - \frac{n \varepsilon^2}{2304 K^2 \alpha^2}\Big) + \exp \Big(- \frac{n\varepsilon^2}{  2 \big(K_Y + \| \beta_0 \| e^{K_X}\big)^2}\Big) \leq 37 \exp \left(- \kappa_4 n \varepsilon^2 \right),
	\end{align*}
	where 
	\[\kappa_4=\min \bigg(\frac{ 1}{2304 K^2 \alpha^2} ,\frac{1}{ 2 \big(K_Y + \| \beta_0 \| e^{K_X}\big)^2}\bigg).\]
	The same analysis can be done to $(-Z_{m,n}(\beta))$, and so $\prob \left( | \widehat{L}_n(m) - L(m) |> \varepsilon \right)\leq 74 \exp \left(- \kappa_4 n \varepsilon^2  \right).$ Moreover, taking $\beta_0=0$ gives
	\[\kappa_4=\min \left(\frac{ 1}{2304 K^2 \alpha^2} ,\frac{1}{ 2 \big(K_Y + \| \beta_0 \| e^{K_X}\big)^2}\right) \geq \frac{1}{2( 1152 K^2 \alpha^2 + K_Y^2)}=C_4, \]
	which completes the proof.
\end{proof}

This allows us to treat the case $m<m^\ast$.
\begin{proposition}
	\label{prop:cvg_hatm_leq_mast}
	Let $0<\rho< \frac{1}{2}$ and $\textnormal{pen}_n(m) $ be defined by \eqref{eq:def_pen}. Let $n_3$ be the smallest integer satisfying
	\begin{equation}
	\label{eq:condition_n_3}
	n_3 \geq \bigg(\frac{2 \sqrt{s_d(m^\ast)}}{L(m^\ast -1) - \sigma^2}\big(432 K \alpha \sqrt{\pi} + K_{\textnormal{pen}}\big)\bigg)^{1/\rho}.
	\end{equation}
	Then, under the assumptions $(H_\alpha)$ and $(H_K)$, for any $m<m^\ast$, $n \geq n_3$,
	\begin{equation*}
	\label{eq:bound_prob_hatm_eq_k_leq_mast}
	\prob\left( \widehat{m} = m \right) \leq 148\exp \Big(-n \frac{C_4}{4}  \big(L(m) - L(m^\ast) -\textnormal{pen}_n(m^\ast) + \textnormal{pen}_n(m)\big)^2\Big),
	\end{equation*}
	where the constant $C_4$ is defined by \eqref{eq:def_C_4}.
\end{proposition}
\begin{proof}[\textbf{\upshape Proof:}]
	This is a consequence of Proposition \ref{prop:cvg_hatL}. For any $m < m^\ast$, 
	\begin{equation*}
	\begin{split}
	\prob (\widehat{m}=m)& \leq \prob  \left(\widehat{L}_n(m) -\widehat{L}_n(m^\ast) \leq \text{pen}_n(m^\ast) -\text{pen}_n(m) \right)\\
	& = \prob  \Big(\widehat{L}_n(m^\ast) - L(m^\ast) + L(m) -\widehat{L}_n(m)\geq L(m) - L(m^\ast) -\big( \text{pen}_n(m^\ast) -\text{pen}_n(m) \big) \Big) \\
	& \leq  \prob \Big(\big|\widehat{L}_n(m) - L(m)\big| \geq \frac{ 1}{2} \big(L(m) - L(m^\ast)- \text{pen}_n(m^\ast) +\text{pen}_n(m) \big) \Big)\\
	&\quad +  \prob \Big( \big|\widehat{L}_n(m^\ast) - L(m^\ast)\big| \geq \frac{ 1}{2} \big( L(m) - L(m^\ast)- \text{pen}_n(m^\ast) +\text{pen}_n(m) \big)\Big). \\
	\end{split}
	\end{equation*}
	In order to apply Proposition \ref{prop:cvg_hatL}, we first need to ensure that $L(m) - L(m^\ast)- \text{pen}_n(m^\ast) +\text{pen}_n(m)$ is strictly positive. Recall that $m \mapsto L(m)$ is a decreasing function, minimal at $m=m^\ast$ and then bounded by $\sigma^2$. Recall also that $m \mapsto \text{pen}_n(m)$ is strictly increasing. This gives, for $m<m^\ast$:
	\begin{equation*}
	L(m) - L(m^\ast)- \text{pen}_n(m^\ast) +\text{pen}_n(m)  > L(m^\ast -1) - \sigma^2 - K_{\textnormal{pen}}n^{-\rho} \sqrt{s_d(m^\ast)}.
	\end{equation*}
	This implies that it is enough that 
	\begin{equation}
	\label{eq:equation_1_n_2}
	L(m^\ast -1) - \sigma^2 - K_{\textnormal{pen}}n^{-\rho} \sqrt{s_d(m^\ast)} > \frac{1}{2} (L(m^\ast-1) - \sigma^2)
	\end{equation}
	 to ensure that $  L(m) - L(m^\ast)- \text{pen}_n(m^\ast) +\text{pen}_n(m)  > 0$. This yields a first condition on $n_3$:
	\begin{equation}
	\label{eq:condition_1_n_3}
	n_3 \geq \bigg( \frac{2 K_{\textnormal{pen}} \sqrt{s_d(m^\ast)}}{L(m^\ast -1) - \sigma^2}\bigg)^{\frac{1}{\rho}}.
	\end{equation}
	However, to apply Proposition \ref{prop:cvg_hatL}, we also need $n_3$ to satisfy \eqref{eq:condition_n_2} , which writes
	\[n_3 \geq \frac{432^2 K^2 \pi \alpha^2 s_d(m)}{( L(m) - L(m^\ast)- \text{pen}_n(m^\ast) +\text{pen}_n(m))^2}. \]
	If $n_3$ satisfies \eqref{eq:condition_1_n_3}, we can bound the right-hand side uniformly in $m$:
	\begin{align*}
	 \frac{432^2 K^2 \pi \alpha^2 s_d(m)}{\big( L(m) - L(m^\ast)- \text{pen}_n(m^\ast) +\text{pen}_n(m) \big)^2 } &\leq \frac{4 \times432^2 K^2 \pi \alpha^2 s_d(m^\ast)}{( L(m^\ast -1) - \sigma^2)^2} =\bigg( \frac{2 \times432 K \alpha\sqrt{\pi s_d(m^\ast)} }{ L(m^\ast -1) - \sigma^2} \bigg)^2.
	 \end{align*}
	We can assume that this quantity is larger than $1$, as otherwise the condition on $n_3$ will be trivially satisfied. Then, as $\rho < \frac{1}{2}$, it is enough for $n_3$ to satisfy
	\[n_3 \geq \max \bigg( \frac{2 K_{\textnormal{pen}} \sqrt{s_d(m^\ast)}}{L(m^\ast -1) -\sigma^2}, \frac{2 \times432 K \alpha\sqrt{\pi s_d(m^\ast)} }{ L(m^\ast -1) - \sigma^2} \bigg)^{1/\rho} , \]
	or in a more compact form that
	\[n_3 \geq \bigg( \frac{2 (K_{\textnormal{pen}} +432 K \alpha\sqrt{\pi }) \sqrt{s_d(m^\ast)}}{L(m^\ast -1) - \sigma^2}\bigg)^{1/\rho} .\]
	We conclude by applying Proposition \ref{prop:cvg_hatL} to both terms with 
	\[\varepsilon= \frac{1}{2} \big( L(m) - L(m^\ast)-\text{pen}_n(m^\ast) -\text{pen}_n(m) \big).\]
\end{proof}

We are now in a position to prove Theorem \ref{th:main_bound_hatm}.
\begin{proof}[\textbf{\upshape Proof:}][Proof of Theorem \ref{th:main_bound_hatm}]
	The result is a consequence of Propositions \ref{prop:cvg_hatm_geq_mast} and \ref{prop:cvg_hatm_leq_mast}. For this, we first need to ensure that the conditions on $n$ \eqref{eq:condition_n_1} and \eqref{eq:condition_n_3} are satisfied. Thus, we need to bound
	\begin{align*}
	M=\max \Bigg( & \bigg(\frac{2 \sqrt{s_d(m^\ast)}}{L(m^\ast -1) - \sigma^2}\big(432 K \alpha \sqrt{\pi} + K_{\textnormal{pen}}\big)\bigg)^{1/\rho}, \bigg( \frac{432 \sqrt{\pi} K \alpha\sqrt{s_d(m^\ast+1)}}{K_{\textnormal{pen}}(\sqrt{s_d(m^\ast+1)} - \sqrt{s_d(m^\ast) }) } \bigg)^{1/(\frac{1}{2} - \rho)} \Bigg) .
	\end{align*}
	If $\tilde{\rho}=\min(\rho, \frac{1}{2}-\rho)$, then
	\begin{align*}
	M& \leq \bigg( (432 K \alpha \sqrt{\pi} + K_{\textnormal{pen}})\sqrt{s_d(m^\ast+1)} \max \bigg(\frac{2}{L(m^\ast -1) - \sigma^2}, \frac{1}{K_{\textnormal{pen}}\big(\sqrt{s_d(m^\ast+1)} - \sqrt{s_d(m^\ast) } \big) }\bigg) \bigg)^{1/\tilde{\rho}} \\
	& \leq \bigg( (432 K \alpha \sqrt{\pi} + K_{\textnormal{pen}})\sqrt{s_d(m^\ast+1)} \Big( \frac{2}{L(m^\ast -1) - \sigma^2}+ \frac{\sqrt{2}}{K_{\textnormal{pen}}\sqrt{d^{m^\ast+1}}} \Big) \bigg)^{1/\tilde{\rho}}.\\
	\end{align*}
	Therefore, condition \eqref{eq:condition_n_0_thm} implies that \eqref{eq:condition_n_1} and \eqref{eq:condition_n_3} are satisfied. Splitting the probability $\prob \left(\widehat{m} \neq m^\ast \right) $ into two terms now gives
	\[\prob \left(\widehat{m} \neq m^\ast \right) = \prob \left(\widehat{m} > m^\ast\right) + \prob \left(\widehat{m} < m^\ast \right) \leq \sum_{m>m^\ast} \prob \left(\widehat{m} =m\right) + \sum_{m<m^\ast} \prob \left(\widehat{m} =m\right) .\]
	On the one hand, Theorem \ref{prop:cvg_hatm_geq_mast} shows that, for $n \geq n_0$,
	\[  \sum_{m>m^\ast} \prob \left(\widehat{m} =m\right)  \leq 74 e^{- C_3 n^{1-2\rho}} \sum_{m>m^\ast} e^{- C_3 s_d(m)},\]
	 and, on the other hand, Proposition \ref{prop:cvg_hatm_leq_mast} gives
	 \begin{align*}
	 \sum_{m<m^\ast} \prob \left(\widehat{m} =m\right) &\leq  148 \sum_{m=0}^{m^\ast -1} \exp \Big(- \frac{C_4}{4}n \big(L(m) - L(m^\ast) - \textnormal{pen}_n(m^\ast) + \textnormal{pen}_n(m) \big) \Big) \leq 148 m^\ast\exp \big(- \frac{C_4}{8}n( L(m^\ast -1) - \sigma^2) \big),
	 \end{align*}
where we have used that for $n \geq n_0$, \eqref{eq:equation_1_n_2} is true. Letting 
\[\kappa_5 = \min \Big(C_3, \frac{C_4(L(m^\ast-1) - \sigma^2)}{8} \Big)\]
 yields
	\begin{align*}
	\prob \left(\widehat{m} \neq m^\ast \right) & \leq 74 e^{-\kappa_5 n^{1-2\rho} }  \sum_{m>0} e^{- C_3s_d(m)}  + 148m^\ast e^{-\kappa_5 n} \leq C_1 e^{- \kappa_5 n^{1-2\rho} },
	\end{align*}
	where 
	\begin{equation*}
	C_1 =  74 \sum_{m>0} e^{-C_3s_d(m)} + 148 m^\ast.
	\end{equation*}
	 To complete the proof, it remains to find a lower bound on $\kappa_5$:
	\begin{align} \label{eq:C_2_smaller_C_3}
	\kappa_5 = \min \Big(C_3, \frac{C_4(L(m^\ast-1) - \sigma^2)}{8}\Big)& = \min\bigg(  \frac{K_{\textnormal{pen}}^2 d^{m^\ast+1}}{128s_d(m^\ast+1)(72K^2\alpha^2 + K_Y^2)},\frac{L(m^\ast -1) - \sigma^2}{16(1152 K^2 \alpha^2 + K_Y^2)} \bigg) \nonumber \\
	& \geq \frac{1}{16(1152 K^2 \alpha^2 + K_Y^2)} \min \Big(\frac{K_{\textnormal{pen}}^2  d^{m^\ast+1}}{8s_d(m^\ast+1)}, L(m^\ast-1) - \sigma^2\Big) =C_2.
	\end{align}
\end{proof}

\begin{flushleft}
\textbf{Proof of Corollary \ref{cor:estimator_hat_beta}}
\end{flushleft}
\label{sec:appendix_proof_consistency_hatbeta}

First, let us note that $\esp \big( \langle \widehat{\beta}_{\widehat{m}},S^{\widehat{m}}(X) \rangle  - \langle \beta^\ast_{m^\ast},S^{m^\ast}(X) \rangle \big)^2 = \esp \big( \risk_{\widehat{m}}(\widehat{\beta}_{\widehat{m}}) \big) - \risk_{m^\ast}(\beta^\ast_{m^\ast}).$ Moreover, we have a.s.
\begin{align*}
\risk_{\widehat{m}}(\widehat{\beta}_{\widehat{m}})  - \risk_{m^\ast}(\beta^\ast_{m^\ast})&= \risk_{\widehat{m}}(\widehat{\beta}_{\widehat{m}}) - \risk_{\widehat{m}}(\beta^\ast_{\widehat{m}}) + \risk_{\widehat{m}}(\beta^\ast_{\widehat{m}}) - \risk_{m^\ast}(\beta^\ast_{m^\ast}) \\
& = \risk_{\widehat{m}}(\widehat{\beta}_{\widehat{m}}) - \widehat{\risk}_{\widehat{m},n}(\widehat{\beta}_{\widehat{m}}) + \widehat{\risk}_{\widehat{m},n}(\widehat{\beta}_{\widehat{m}}) -\widehat{\risk}_{\widehat{m},n}(\beta^\ast_{\widehat{m}}) +\widehat{\risk}_{\widehat{m},n}(\beta^\ast_{\widehat{m}}) - \risk_{\widehat{m}}(\beta^\ast_{\widehat{m}}) + \risk_{\widehat{m}}(\beta^\ast_{\widehat{m}}) - \risk_{m^\ast}(\beta^\ast_{m^\ast})  \\
& \leq \risk_{\widehat{m}}(\widehat{\beta}_{\widehat{m}}) - \widehat{\risk}_{\widehat{m},n}(\widehat{\beta}_{\widehat{m}}) +\widehat{\risk}_{\widehat{m},n}(\beta^\ast_{\widehat{m}}) - \risk_{\widehat{m}}(\beta^\ast_{\widehat{m}}) + \risk_{\widehat{m}}(\beta^\ast_{\widehat{m}}) - \risk_{m^\ast}(\beta^\ast_{m^\ast}) \\
& \leq 2 \sup_{\beta \in B_{\widehat{m},\alpha}} | \widehat{\risk}_{\widehat{m},n}(\beta) - \risk_{\widehat{m}}(\beta) | + \risk_{\widehat{m}}(\beta^\ast_{\widehat{m}}) - \risk_{m^\ast}(\beta^\ast_{m^\ast}).
\end{align*}

We decompose the proof into two lemmas.

\begin{lemma}
\[ \esp \big[ \sup_{\beta \in B_{\widehat{m},\alpha}} | \widehat{\risk}_{\widehat{m},n}(\beta) - \risk_{\widehat{m}}(\beta) |  \big] \leq 36 K \alpha \sqrt{\frac{\pi}{n}} \Big((m^\ast +1) \sqrt{s_d(m^\ast)} + 74e^{-C_3 n^{1-2\rho }} \sum_{m>m^\ast} \sqrt{s_d(m)} e^{-C_3 s_d(m)} \Big), \]
where the constant $C_3$ is defined by \eqref{eq:def_C_3}.
\end{lemma}

\begin{proof}[\textbf{\upshape Proof:}]
	From Corollary 5.25 of \cite{van2014probability} and \eqref{eq:entropy_integral}, for any $m \in \N$,
	\begin{align*}
	\esp \Big( \sup_{\beta \in B_{m,\alpha}} | \widehat{\risk}_{m,n}(\beta) - \risk_m(\beta) |  \Big) &\leq 12 \int_{0}^{\infty} \sqrt{\log(N(B_{m,\alpha},D,\varepsilon))} = 36 K \alpha \sqrt{s_d(m)} \sqrt{\frac{\pi}{n}},
	\end{align*}
	where $N(B_{m,\alpha},D,\varepsilon)$ is the $\varepsilon$-covering number of $B_{m,\alpha}$ with respect to the distance $D$, defined by \eqref{eq:def_D}. This gives, for $m=\widehat{m}$,
	\[ \esp \Big( \sup_{\beta \in B_{\widehat{m},\alpha}} | \widehat{\risk}_{\widehat{m},n}(\beta) - \risk_{\widehat{m}}(\beta) |  \Big) \leq 36 K \alpha \sqrt{\frac{\pi}{n}} \esp \Big(  \sqrt{s_d(\widehat{m})} \Big).  \]
	To compute this expectation, Proposition \ref{prop:cvg_hatm_geq_mast} yields
	\begin{align*}
	\esp \Big( \sqrt{s_d(\widehat{m})} \Big)& = \sum_{m\leq m^\ast} \sqrt{s_d(m)} \prob(\widehat{m}=m) + \sum_{m>m^\ast} \sqrt{s_d(m)} \prob(\widehat{m}=m) \\
	&\leq (m^\ast +1)\sqrt{s_d(m^\ast)} + \sum_{m>m^\ast} \sqrt{s_d(m)} 74 \exp \big(-C_3 (n^{1-2\rho} +s_d(m) ) \big) \\
	& \leq (m^\ast +1)\sqrt{s_d(m^\ast)} + e^{-C_3 n^{1-2\rho}}\sum_{m>m^\ast} \sqrt{s_d(m)} 74 \exp \big(-C_3 s_d(m) \big),
	\end{align*}
	which completes the proof.
\end{proof}

\begin{lemma}
	$\esp \big( \risk_{\widehat{m}}(\beta^\ast_{\widehat{m}}) - \risk_{m^\ast}(\beta^\ast_{m^\ast}) \big) =2 \alpha^2 e^{K_X}  C_1 e^{- C_2n^{1-2\rho}},$ where the constants $C_1$ and $C_2$ are defined by \eqref{eq:def_C_1} and \eqref{eq:def_C_2}.
\end{lemma}

\begin{proof}[\textbf{\upshape Proof:}]
	Since, for any $m\in \N$, $\langle \beta^\ast_{m}, S^{m}(X)\rangle^2 \leq \| \beta^\ast_k\|^2_2 \| S^m(X) \|^2_2 \leq \alpha^2 e^{K_X},$
	it follows that
	\begin{align*}
	\esp \big( \risk_{\widehat{m}}(\beta^\ast_{\widehat{m}}) - \risk_{m^\ast}(\beta^\ast_{m^\ast}) \big) &= \esp \Big( \big(Y-\langle \beta^\ast_{\widehat{m}}, S^{\widehat{m}}(X) \rangle \big)^2 -  \big(Y-\langle \beta^\ast_{m^\ast}, S^{m^\ast}(X) \rangle \big)^2  \Big) \\
	& = \esp \Big( \big(\langle \beta^\ast_{m^\ast}, S^{m^\ast}(X) \rangle + \varepsilon -\langle \beta^\ast_{\widehat{m}}, S^{\widehat{m}}(X)\rangle \big)^2  - \varepsilon^2 \Big) \\
	& = \esp\Big( \big(\langle \beta^\ast_{m^\ast}, S^{m^\ast}(X) \rangle -\langle \beta^\ast_{\widehat{m}}, S^{\widehat{m}}(X)\rangle \big)^2 \Big)\leq 2 \alpha^2 e^{K_X} \prob(\widehat{m} \neq m^\ast).
	\end{align*}

	By Theorem \ref{th:main_bound_hatm}, this yields $\esp \big( \risk_{\widehat{m}}(\beta^\ast_{\widehat{m}}) - \risk_{m^\ast}(\beta^\ast_{m^\ast}) \big) \leq 2 \alpha^2 e^{K_X}  C_1 e^{- C_2n^{1-2\rho}},$ where $C_1$ and $C_2$ are defined by \eqref{eq:def_C_1} and \eqref{eq:def_C_2}.
\end{proof}
Letting $C_5 = 36 K \alpha \sqrt{\pi}(m^\ast +1) \sqrt{s_d(m^\ast)}$, and $C_6 = 2664 K \alpha \sqrt{\pi} \sum_{m>m^\ast} \sqrt{s_d(m)} e^{-C_3 s_d(m)}+ 2 \alpha^2 e^{K_X}  C_1$, since, by \eqref{eq:C_2_smaller_C_3}, $C_2 \leq C_3$, we conclude that
\begin{align*}
\esp \big( \langle \widehat{\beta}_{\widehat{m}},S^{\widehat{m}}(X) \rangle  - \langle \beta^\ast_{m^\ast},S^{m^\ast}(X) \rangle \big)^2 \leq \frac{C_5}{\sqrt{n}} + C_6 e^{-C_2 n^{1-2\rho}}.
\end{align*}

\section*{Acknowledgments}

This work was supported by a grant from R\'egion Ile-de-France. I would like to thank Gérard Biau (Sorbonne Université) and Benoît Cadre (Université Rennes 2) for stimulating discussions and insightful suggestions. I also thank the Editor and two anonymous referees for their careful reading of the paper and constructive comments, which led to a substantial improvement of the article.

\bibliographystyle{myjmva}
\bibliography{references}

\begin{thebibliography}{46}
\expandafter\ifx\csname natexlab\endcsname\relax\def\natexlab#1{#1}\fi
\providecommand{\bibinfo}[2]{#2}
\ifx\xfnm\relax \def\xfnm[#1]{\unskip,\space#1}\fi
\bibitem[{Arlot(2019)}]{arlot2019minimal}
\bibinfo{author}{S.~Arlot}, \bibinfo{title}{Minimal penalties and the slope
  heuristics: a survey}, \bibinfo{journal}{Journal de la Société Française
  de Statistique} \bibinfo{volume}{160} (\bibinfo{year}{2019})
  \bibinfo{pages}{1--106}.
\bibitem[{Arribas et~al.(2018)Arribas, Goodwin, Geddes, Lyons and
  Saunders}]{arribas2018signature}
\bibinfo{author}{I.~P. Arribas}, \bibinfo{author}{G.~M. Goodwin},
  \bibinfo{author}{J.~R. Geddes}, \bibinfo{author}{T.~Lyons},
  \bibinfo{author}{K.~E. Saunders}, \bibinfo{title}{A signature-based machine
  learning model for distinguishing bipolar disorder and borderline personality
  disorder}, \bibinfo{journal}{Translational psychiatry} \bibinfo{volume}{8}
  (\bibinfo{year}{2018}) \bibinfo{pages}{1--7}.
\bibitem[{Arribas et~al.(2020)Arribas, Salvi and Szpruch}]{arribas2020sig}
\bibinfo{author}{I.~P. Arribas}, \bibinfo{author}{C.~Salvi},
  \bibinfo{author}{L.~Szpruch}, \bibinfo{title}{{Sig-SDEs} model for
  quantitative finance}, \bibinfo{journal}{arXiv:2006.00218}
  (\bibinfo{year}{2020}).
\bibitem[{Benzeghiba et~al.(2007)Benzeghiba, De~Mori, Deroo, Dupont, Erbes,
  Jouvet, Fissore, Laface, Mertins, Ris et~al.}]{benzeghiba2007automatic}
\bibinfo{author}{M.~Benzeghiba}, \bibinfo{author}{R.~De~Mori},
  \bibinfo{author}{O.~Deroo}, \bibinfo{author}{S.~Dupont},
  \bibinfo{author}{T.~Erbes}, \bibinfo{author}{D.~Jouvet},
  \bibinfo{author}{L.~Fissore}, \bibinfo{author}{P.~Laface},
  \bibinfo{author}{A.~Mertins}, \bibinfo{author}{C.~Ris}, et~al.,
  \bibinfo{title}{Automatic speech recognition and speech variability: A
  review}, \bibinfo{journal}{Speech communication} \bibinfo{volume}{49}
  (\bibinfo{year}{2007}) \bibinfo{pages}{763--786}.
\bibitem[{Birg{\'e} and Massart(2007)}]{birge2007minimal}
\bibinfo{author}{L.~Birg{\'e}}, \bibinfo{author}{P.~Massart},
  \bibinfo{title}{Minimal penalties for gaussian model selection},
  \bibinfo{journal}{Probability Theory and Related Fields}
  \bibinfo{volume}{138} (\bibinfo{year}{2007}) \bibinfo{pages}{33--73}.
\bibitem[{Brunel et~al.(2016)Brunel, Mas and Roche}]{brunel2016non}
\bibinfo{author}{{\'E}.~Brunel}, \bibinfo{author}{A.~Mas},
  \bibinfo{author}{A.~Roche}, \bibinfo{title}{Non-asymptotic adaptive
  prediction in functional linear models}, \bibinfo{journal}{Journal of
  Multivariate Analysis} \bibinfo{volume}{143} (\bibinfo{year}{2016})
  \bibinfo{pages}{208--232}.
\bibitem[{Cardot et~al.(1999)Cardot, Ferraty and Sarda}]{cardot1999functional}
\bibinfo{author}{H.~Cardot}, \bibinfo{author}{F.~Ferraty},
  \bibinfo{author}{P.~Sarda}, \bibinfo{title}{Functional linear model},
  \bibinfo{journal}{Statistics \& Probability Letters} \bibinfo{volume}{45}
  (\bibinfo{year}{1999}) \bibinfo{pages}{11--22}.
\bibitem[{Cardot et~al.(2003)Cardot, Ferraty and Sarda}]{cardot2003spline}
\bibinfo{author}{H.~Cardot}, \bibinfo{author}{F.~Ferraty},
  \bibinfo{author}{P.~Sarda}, \bibinfo{title}{Spline estimators for the
  functional linear model}, \bibinfo{journal}{Statistica Sinica}
  (\bibinfo{year}{2003}) \bibinfo{pages}{571--591}.
\bibitem[{Chen(1958)}]{chen1958integration}
\bibinfo{author}{K.-T. Chen}, \bibinfo{title}{Integration of paths—a faithful
  representation of paths by non-commutative formal power series},
  \bibinfo{journal}{Transactions of the American Mathematical Society}
  \bibinfo{volume}{89} (\bibinfo{year}{1958}) \bibinfo{pages}{395--407}.
\bibitem[{Chevyrev and Kormilitzin(2016)}]{chevyrev2016primer}
\bibinfo{author}{I.~Chevyrev}, \bibinfo{author}{A.~Kormilitzin},
  \bibinfo{title}{A primer on the signature method in machine learning},
  \bibinfo{journal}{arXiv:1603.03788}  (\bibinfo{year}{2016}).
\bibitem[{De~Vito et~al.(2008)De~Vito, Massera, Piga, Martinotto and
  Di~Francia}]{de2008field}
\bibinfo{author}{S.~De~Vito}, \bibinfo{author}{E.~Massera},
  \bibinfo{author}{M.~Piga}, \bibinfo{author}{L.~Martinotto},
  \bibinfo{author}{G.~Di~Francia}, \bibinfo{title}{On field calibration of an
  electronic nose for benzene estimation in an urban pollution monitoring
  scenario}, \bibinfo{journal}{Sensors and Actuators B: Chemical}
  \bibinfo{volume}{129} (\bibinfo{year}{2008}) \bibinfo{pages}{750--757}.
\bibitem[{Fermanian(2021)}]{fermanian2021embedding}
\bibinfo{author}{A.~Fermanian}, \bibinfo{title}{Embedding and learning with
  signatures}, \bibinfo{journal}{Computational Statistics \& Data Analysis}
  \bibinfo{volume}{157} (\bibinfo{year}{2021}) \bibinfo{pages}{107148}.
\bibitem[{Ferraty and Vieu(2006)}]{ferraty2006nonparametric}
\bibinfo{author}{F.~Ferraty}, \bibinfo{author}{P.~Vieu},
  \bibinfo{title}{Nonparametric Functional Data Analysis: Theory and Practice},
  \bibinfo{publisher}{Springer}, \bibinfo{address}{New York},
  \bibinfo{year}{2006}.
\bibitem[{Frank and Friedman(1993)}]{frank1993statistical}
\bibinfo{author}{L.~E. Frank}, \bibinfo{author}{J.~H. Friedman},
  \bibinfo{title}{A statistical view of some chemometrics regression tools},
  \bibinfo{journal}{Technometrics} \bibinfo{volume}{35} (\bibinfo{year}{1993})
  \bibinfo{pages}{109--135}.
\bibitem[{Friz and Victoir(2010)}]{friz2010multidimensional}
\bibinfo{author}{P.~K. Friz}, \bibinfo{author}{N.~B. Victoir},
  \bibinfo{title}{Multidimensional Stochastic Processes as Rough Paths: Theory
  and Applications}, volume \bibinfo{volume}{120} of
  \text{\bibinfo{series}{Cambridge Studies in Advanced Mathematics}},
  \bibinfo{publisher}{Cambridge University Press},
  \bibinfo{address}{Cambridge}, \bibinfo{year}{2010}.
\bibitem[{Giusti and Lee(2020)}]{giusti2020iterated}
\bibinfo{author}{C.~Giusti}, \bibinfo{author}{D.~Lee}, \bibinfo{title}{Iterated
  integrals and population time series analysis}, in:
  \bibinfo{booktitle}{Topological Data Analysis},
  \bibinfo{publisher}{Springer}, \bibinfo{year}{2020}, pp.
  \bibinfo{pages}{219--246}.
\bibitem[{Greven et~al.(2011)Greven, Crainiceanu, Caffo and
  Reich}]{greven2011longitudinal}
\bibinfo{author}{S.~Greven}, \bibinfo{author}{C.~Crainiceanu},
  \bibinfo{author}{B.~Caffo}, \bibinfo{author}{D.~Reich},
  \bibinfo{title}{Longitudinal functional principal component analysis}, in:
  \bibinfo{booktitle}{Recent Advances in Functional Data Analysis and Related
  Topics}, \bibinfo{publisher}{Springer}, \bibinfo{year}{2011}, pp.
  \bibinfo{pages}{149--154}.
\bibitem[{Hall et~al.(2007)Hall, Horowitz et~al.}]{hall2007methodology}
\bibinfo{author}{P.~Hall}, \bibinfo{author}{J.~L. Horowitz}, et~al.,
  \bibinfo{title}{Methodology and convergence rates for functional linear
  regression}, \bibinfo{journal}{The Annals of Statistics} \bibinfo{volume}{35}
  (\bibinfo{year}{2007}) \bibinfo{pages}{70--91}.
\bibitem[{Hambly and Lyons(2010)}]{hambly2010uniqueness}
\bibinfo{author}{B.~Hambly}, \bibinfo{author}{T.~Lyons},
  \bibinfo{title}{Uniqueness for the signature of a path of bounded variation
  and the reduced path group}, \bibinfo{journal}{The Annals of Mathematics}
  \bibinfo{volume}{171} (\bibinfo{year}{2010}) \bibinfo{pages}{109--167}.
\bibitem[{van Handel(2014)}]{van2014probability}
\bibinfo{author}{R.~van Handel}, \bibinfo{title}{Probability in high
  dimension}, \bibinfo{type}{Technical Report}, Princeton University,
  \bibinfo{year}{2014}.
\bibitem[{Hastie and Mallows(1993)}]{hastie1993statistical}
\bibinfo{author}{T.~Hastie}, \bibinfo{author}{C.~Mallows}, \bibinfo{title}{[a
  statistical view of some chemometrics regression tools]: Discussion},
  \bibinfo{journal}{Technometrics} \bibinfo{volume}{35} (\bibinfo{year}{1993})
  \bibinfo{pages}{140--143}.
\bibitem[{Kir{\'{a}}ly and Oberhauser(2019)}]{kiraly2016kernels}
\bibinfo{author}{F.~J. Kir{\'{a}}ly}, \bibinfo{author}{H.~Oberhauser},
  \bibinfo{title}{{Kernels for sequentially ordered data}},
  \bibinfo{journal}{Journal of Machine Learning Research} \bibinfo{volume}{20}
  (\bibinfo{year}{2019}) \bibinfo{pages}{1--45}.
\bibitem[{Lai et~al.(2017)Lai, Jin and Yang}]{lai2017online}
\bibinfo{author}{S.~Lai}, \bibinfo{author}{L.~Jin}, \bibinfo{author}{W.~Yang},
  \bibinfo{title}{Online signature verification using recurrent neural network
  and length-normalized path signature descriptor}, in:
  \bibinfo{booktitle}{Proceedings of the 14th IAPR International Conference on
  Document Analysis and Recognition (ICDAR)}, volume~\bibinfo{volume}{1},
  \bibinfo{organization}{IEEE}, pp. \bibinfo{pages}{400--405}.
\bibitem[{Le~Jan and Qian(2013)}]{le2013stratonovich}
\bibinfo{author}{Y.~Le~Jan}, \bibinfo{author}{Z.~Qian},
  \bibinfo{title}{Stratonovich’s signatures of brownian motion determine
  brownian sample paths}, \bibinfo{journal}{Probability Theory and Related
  Fields} \bibinfo{volume}{157} (\bibinfo{year}{2013})
  \bibinfo{pages}{209--223}.
\bibitem[{Levin et~al.(2013)Levin, Lyons and Ni}]{levin2013learning}
\bibinfo{author}{D.~Levin}, \bibinfo{author}{T.~Lyons},
  \bibinfo{author}{H.~Ni}, \bibinfo{title}{Learning from the past, predicting
  the statistics for the future, learning an evolving system},
  \bibinfo{journal}{arXiv:1309.0260}  (\bibinfo{year}{2013}).
\bibitem[{Li et~al.(2017)Li, Zhang and Jin}]{li2017lpsnet}
\bibinfo{author}{C.~Li}, \bibinfo{author}{X.~Zhang}, \bibinfo{author}{L.~Jin},
  \bibinfo{title}{{LPSNet: a novel log path signature feature based hand
  gesture recognition framework}}, in: \bibinfo{booktitle}{2017 IEEE
  International Conference on Computer Vision Workshop}, pp.
  \bibinfo{pages}{631--639}.
\bibitem[{Li and Hsing(2007)}]{li2007rates}
\bibinfo{author}{Y.~Li}, \bibinfo{author}{T.~Hsing}, \bibinfo{title}{On rates
  of convergence in functional linear regression}, \bibinfo{journal}{Journal of
  Multivariate Analysis} \bibinfo{volume}{98} (\bibinfo{year}{2007})
  \bibinfo{pages}{1782--1804}.
\bibitem[{Liu et~al.(2017)Liu, Jin and Xie}]{liu2017ps}
\bibinfo{author}{M.~Liu}, \bibinfo{author}{L.~Jin}, \bibinfo{author}{Z.~Xie},
  \bibinfo{title}{Ps-lstm: Capturing essential sequential online information
  with path signature and lstm for writer identification}, in:
  \bibinfo{booktitle}{Proceedings of the 14th IAPR International Conference on
  Document Analysis and Recognition (ICDAR)}, volume~\bibinfo{volume}{1},
  \bibinfo{organization}{IEEE}, pp. \bibinfo{pages}{664--669}.
\bibitem[{Lyons(2014)}]{lyons2014rough}
\bibinfo{author}{T.~Lyons}, \bibinfo{title}{Rough paths, signatures and the
  modelling of functions on streams}, \bibinfo{journal}{arXiv:1405.4537}
  (\bibinfo{year}{2014}).
\bibitem[{Lyons et~al.(2007)Lyons, Caruana and
  L{\'e}vy}]{lyons2007differential}
\bibinfo{author}{T.~Lyons}, \bibinfo{author}{M.~Caruana},
  \bibinfo{author}{T.~L{\'e}vy}, \bibinfo{title}{Differential Equations driven
  by Rough Paths}, volume \bibinfo{volume}{1908} of
  \text{\bibinfo{series}{Lecture Notes in Mathematics}},
  \bibinfo{publisher}{Springer}, \bibinfo{address}{Berlin},
  \bibinfo{year}{2007}.
\bibitem[{Marx and Eilers(1999)}]{marx1999generalized}
\bibinfo{author}{B.~D. Marx}, \bibinfo{author}{P.~H. Eilers},
  \bibinfo{title}{Generalized linear regression on sampled signals and curves:
  a {P}-spline approach}, \bibinfo{journal}{Technometrics} \bibinfo{volume}{41}
  (\bibinfo{year}{1999}) \bibinfo{pages}{1--13}.
\bibitem[{Moore et~al.(2019)Moore, Lyons and Gallacher}]{moore2019using}
\bibinfo{author}{P.~Moore}, \bibinfo{author}{T.~Lyons},
  \bibinfo{author}{J.~Gallacher}, \bibinfo{title}{Using path signatures to
  predict a diagnosis of {Alzheimer’s} disease}, \bibinfo{journal}{PloS ONE}
  \bibinfo{volume}{14} (\bibinfo{year}{2019}).
\bibitem[{Morrill et~al.(2020{\natexlab{a}})Morrill, Fermanian, Kidger and
  Lyons}]{morrill2020generalised}
\bibinfo{author}{J.~Morrill}, \bibinfo{author}{A.~Fermanian},
  \bibinfo{author}{P.~Kidger}, \bibinfo{author}{T.~Lyons}, \bibinfo{title}{A
  generalised signature method for multivariate time series feature
  extraction}, \bibinfo{journal}{arXiv:2006.00873}
  (\bibinfo{year}{2020}{\natexlab{a}}).
\bibitem[{Morrill et~al.(2019)Morrill, Kormilitzin, Nevado-Holgado,
  Swaminathan, Howison and Lyons}]{morrill2019sepsis}
\bibinfo{author}{J.~Morrill}, \bibinfo{author}{A.~Kormilitzin},
  \bibinfo{author}{A.~Nevado-Holgado}, \bibinfo{author}{S.~Swaminathan},
  \bibinfo{author}{S.~Howison}, \bibinfo{author}{T.~Lyons}, \bibinfo{title}{The
  signature-based model for early detection of sepsis from electronic health
  records in the intensive care unit}, \bibinfo{journal}{International
  Conference in Computing in Cardiology}  (\bibinfo{year}{2019}).
\bibitem[{Morrill et~al.(2020{\natexlab{b}})Morrill, Kormilitzin,
  Nevado-Holgado, Swaminathan, Howison and Lyons}]{morrill2020utilization}
\bibinfo{author}{J.~H. Morrill}, \bibinfo{author}{A.~Kormilitzin},
  \bibinfo{author}{A.~J. Nevado-Holgado}, \bibinfo{author}{S.~Swaminathan},
  \bibinfo{author}{S.~D. Howison}, \bibinfo{author}{T.~J. Lyons},
  \bibinfo{title}{Utilization of the signature method to identify the early
  onset of sepsis from multivariate physiological time series in critical care
  monitoring}, \bibinfo{journal}{Critical Care Medicine} \bibinfo{volume}{48}
  (\bibinfo{year}{2020}{\natexlab{b}}) \bibinfo{pages}{e976--e981}.
\bibitem[{Morris(2015)}]{morris2015functional}
\bibinfo{author}{J.~S. Morris}, \bibinfo{title}{Functional regression},
  \bibinfo{journal}{Annual Review of Statistics and Its Application}
  \bibinfo{volume}{2} (\bibinfo{year}{2015}) \bibinfo{pages}{321--359}.
\bibitem[{Park and Staicu(2015)}]{park2015longitudinal}
\bibinfo{author}{S.~Y. Park}, \bibinfo{author}{A.-M. Staicu},
  \bibinfo{title}{Longitudinal functional data analysis},
  \bibinfo{journal}{Stat} \bibinfo{volume}{4} (\bibinfo{year}{2015})
  \bibinfo{pages}{212--226}.
\bibitem[{Ramos-Carre{\~n}o et~al.(2019)Ramos-Carre{\~n}o, Torrecilla and
  Su{\'a}rez}]{ramos2019scikit}
\bibinfo{author}{C.~Ramos-Carre{\~n}o}, \bibinfo{author}{J.~L. Torrecilla},
  \bibinfo{author}{A.~Su{\'a}rez}, \bibinfo{title}{Scikit-fda: A python package
  for functional data analysis}, in: \bibinfo{booktitle}{3rd International
  Workshop on Advances in Functional Data Analysis},
  volume~\bibinfo{volume}{5}.
\bibitem[{Ramsay and Dalzell(1991)}]{ramsay1991some}
\bibinfo{author}{J.~O. Ramsay}, \bibinfo{author}{C.~Dalzell},
  \bibinfo{title}{Some tools for functional data analysis},
  \bibinfo{journal}{Journal of the Royal Statistical Society. Series B
  (Methodological)} \bibinfo{volume}{53} (\bibinfo{year}{1991})
  \bibinfo{pages}{539--561}.
\bibitem[{Ramsay and Silverman(2005)}]{ramsay2005functional}
\bibinfo{author}{J.~O. Ramsay}, \bibinfo{author}{B.~W. Silverman},
  \bibinfo{title}{Functional Data Analysis. 2nd Edition.},
  \bibinfo{publisher}{Springer}, \bibinfo{address}{New York},
  \bibinfo{year}{2005}.
\bibitem[{Reizenstein and Graham(2020)}]{reizenstein2018iisignature}
\bibinfo{author}{J.~Reizenstein}, \bibinfo{author}{B.~Graham},
  \bibinfo{title}{Algorithm 1004: The iisignature library: Efficient
  calculation of iterated-integral signatures and log signatures},
  \bibinfo{journal}{ACM Transactions on Mathematical Software}
  (\bibinfo{year}{2020}).
\bibitem[{Turaga et~al.(2008)Turaga, Chellappa, Subrahmanian and
  Udrea}]{turaga2008machine}
\bibinfo{author}{P.~Turaga}, \bibinfo{author}{R.~Chellappa},
  \bibinfo{author}{V.~S. Subrahmanian}, \bibinfo{author}{O.~Udrea},
  \bibinfo{title}{Machine recognition of human activities: A survey},
  \bibinfo{journal}{IEEE Transactions on Circuits and Systems for Video
  technology} \bibinfo{volume}{18} (\bibinfo{year}{2008})
  \bibinfo{pages}{1473--1488}.
\bibitem[{Wang et~al.(2019)Wang, Liakata, Ni, Lyons, Nevado-Holgado and
  Saunders}]{wang2019path}
\bibinfo{author}{B.~Wang}, \bibinfo{author}{M.~Liakata},
  \bibinfo{author}{H.~Ni}, \bibinfo{author}{T.~Lyons}, \bibinfo{author}{A.~J.
  Nevado-Holgado}, \bibinfo{author}{K.~Saunders}, \bibinfo{title}{A path
  signature approach for speech emotion recognition}, in:
  \bibinfo{booktitle}{Interspeech 2019}, pp. \bibinfo{pages}{1661--1665}.
\bibitem[{Yang et~al.(2015)Yang, Jin and Liu}]{yang2015chinese}
\bibinfo{author}{W.~Yang}, \bibinfo{author}{L.~Jin}, \bibinfo{author}{M.~Liu},
  \bibinfo{title}{Chinese character-level writer identification using path
  signature feature, dropstroke and deep cnn}, in:
  \bibinfo{booktitle}{Proceedings of the 13th International Conference on
  Document Analysis and Recognition (ICDAR)}, \bibinfo{organization}{IEEE}, pp.
  \bibinfo{pages}{546--550}.
\bibitem[{Yang et~al.(2016)Yang, Jin and Liu}]{yang2016deepwriterid}
\bibinfo{author}{W.~Yang}, \bibinfo{author}{L.~Jin}, \bibinfo{author}{M.~Liu},
  \bibinfo{title}{{{D}eep{W}riter{ID}: {A}n end-to-end online text-independent
  writer identification system}}, \bibinfo{journal}{IEEE Intelligent Systems}
  \bibinfo{volume}{31} (\bibinfo{year}{2016}) \bibinfo{pages}{45--53}.
\bibitem[{Yang et~al.(2017)Yang, Lyons, Ni, Schmid, Jin and
  Chang}]{yang2017leveraging}
\bibinfo{author}{W.~Yang}, \bibinfo{author}{T.~Lyons}, \bibinfo{author}{H.~Ni},
  \bibinfo{author}{C.~Schmid}, \bibinfo{author}{L.~Jin},
  \bibinfo{author}{J.~Chang}, \bibinfo{title}{{Developing the path signature
  methodology and its application to landmark-based human action recognition}},
  \bibinfo{journal}{arXiv:1707.03993}  (\bibinfo{year}{2017}).

\end{thebibliography}

\end{document}